\documentclass[sigconf]{acmart}
\AtBeginDocument{%
  }
\renewcommand\footnotetextcopyrightpermission[1]{} 
\usepackage{makecell}
\usepackage[ruled,vlined]{algorithm2e}
\usepackage{algpseudocode}
\usepackage{listings}
\usepackage{multirow}
\usepackage{subfigure}
\usepackage{tabularx}
\usepackage[table]{xcolor}
\usepackage{fontawesome}
\usepackage{tcolorbox}
\tcbuselibrary{skins}
\usepackage{siunitx}
\usepackage{enumitem}
\setlist[itemize]{leftmargin=*}
\usepackage{pgfplots}
\usepackage{textcomp}
\pgfplotsset{compat=1.18}
\usepgfplotslibrary{groupplots,statistics}
\usetikzlibrary{patterns,positioning,arrows.meta,fit}
\definecolor{cPal}{HTML}{2A78D6}  
\definecolor{cAlt}{HTML}{EB6834}  
\definecolor{cThr}{HTML}{1BAF7A}  
\definecolor{cGrid}{HTML}{DDDDDD}
\definecolor{tabHead}{HTML}{EEF2F7}
\colorlet{tabOurs}{cPal!12}
\colorlet{tabBest}{cThr!22}
\colorlet{tabWeak}{cAlt!14}
\newcommand{\tabstyle}{\footnotesize\renewcommand{\arraystretch}{1.18}\setlength{\aboverulesep}{0pt}\setlength{\belowrulesep}{0pt}}
\newcommand{\best}[1]{\cellcolor{tabBest}\textbf{#1}}
\newcommand{\weak}[1]{\cellcolor{tabWeak}#1}
\newcommand{\tabtop}{\arrayrulecolor{black!80}\specialrule{0.8pt}{0pt}{0pt}\arrayrulecolor{black!30}}
\newcommand{\tabmid}{\arrayrulecolor{black!35}\specialrule{0.4pt}{0pt}{0pt}}
\newcommand{\tabbot}{\arrayrulecolor{black!80}\specialrule{0.8pt}{0pt}{0pt}\arrayrulecolor{black}}
\pgfplotsset{
  paperstyle/.style={
    font=\footnotesize,
    tick label style={font=\footnotesize, text=black!70},
    label style={font=\footnotesize, text=black!85},
    title style={font=\footnotesize\bfseries, text=black!85, yshift=-3pt},
    axis lines*=left,
    axis line style={black!40, line width=0.45pt},
    tick style={black!40, line width=0.45pt},
    major tick length=2pt,
    tick align=outside,
    ymajorgrids,
    major grid style={cGrid!80, line width=0.35pt},
    every axis plot/.append style={line join=round, line cap=round},
    mark options={solid, draw=white, line width=0.4pt},
    legend style={font=\footnotesize, draw=none, fill=none, legend columns=-1, column sep=7pt,
                  /tikz/every even column/.append style={column sep=7pt}},
    legend cell align=left,
  },
  barlegend/.style={legend image code/.code={\fill[##1] (0pt,-2.6pt) rectangle (5.2pt,2.6pt);}},
  linelegend/.style={legend image code/.code={\draw[##1, mark repeat=2, mark phase=2] plot coordinates {(0pt,0pt) (7pt,0pt) (14pt,0pt)};}},
  outbar/.style={draw=white, line width=0.3pt},
}

\newcommand{\Rmnum}[1]{\uppercase\expandafter{\romannumeral #1}}  
\newcommand{\rmnum}[1]{\romannumeral #1} 

\newtcolorbox{takeawaybox}{enhanced, colback=cPal!6, colframe=cPal!80!black, boxrule=0pt, leftrule=2.2pt, arc=0pt, outer arc=0pt, left=4pt, right=4pt, top=2.5pt, bottom=2.5pt, before skip=4pt, after skip=3pt, fontupper=\small}
\newcommand{\takeaway}[1]{\begin{takeawaybox}{\color{cPal!80!black}\faLightbulbO\ \textbf{Takeaway:}} #1\end{takeawaybox}}

\setcopyright{acmlicensed}
\copyrightyear{2018}
\acmYear{2018}
\acmDOI{XXXXXXX.XXXXXXX}
\acmConference[Conference acronym 'XX]{Make sure to enter the correct
  conference title from your rights confirmation email}{June 03--05,
  2018}{Woodstock, NY}
\acmISBN{978-1-4503-XXXX-X/2018/06}

\author{Hao Fu}
\email{haofu@vt.edu}
\affiliation{%
  \institution{Virginia Tech}
  \country{}
}

\author{Jinxia Wei}
\email{weijinxia@cnic.cn}
\affiliation{%
  \institution{Computer Network Information Center, Chinese Academy of Sciences}
  \country{}
}
\author{Wan Wei}
\email{wanwei@cnic.cn}
\affiliation{%
  \institution{Computer Network Information Center, Chinese Academy of Sciences}
  \country{}
}

\author{Chun Long}
\email{longchun@cnic.cn}
\affiliation{%
  \institution{Computer Network Information Center, Chinese Academy of Sciences}
  \country{}
}

\author{Degang Sun}
\email{dgsun@cnic.cn}
\affiliation{%
  \institution{Computer Network Information Center, Chinese Academy of Sciences}
  \country{}
}

\author{Peng Gao}
\email{penggao@vt.edu}
\affiliation{%
  \institution{Virginia Tech}
  \country{}
}

\begin{document}

\title{IDSForge: Self-Explaining Network Intrusion Detection Through Neural Behavioral Rule Learning}

\begin{abstract}
Anomaly-based network intrusion detection systems (NIDS) can detect previously unseen attacks by identifying deviations from benign traffic. However, deep learning-based anomaly detectors typically rely on implicit decision boundaries, and post-hoc explanations may not faithfully reflect the detection logic. This makes alerts difficult to interpret and audit. To address this issue, we present IDSForge, an anomaly-based NIDS that directly learns explicit, executable behavioral rules from benign traffic, making its detection logic interpretable by design. IDSForge uses the same rule set for detection and explanation, ensuring explanation fidelity by construction. 
It uses an explicit distributional autoencoder 
trained only on benign samples to learn axis-aligned regions of benign traffic. 
A high-speed rule compiler then consolidates these instance-specific regions into a compact set of feature-level interval constraints. These rules can be executed directly by popular NIDS engines, supporting high-throughput inspection with low overhead. 
To adapt to changes in traffic, a self-distillation co-evolver updates the model and rule set incrementally when the false-positive rate on analyst-verified traffic exceeds a target threshold. 
Our evaluation shows that IDSForge reduces F1 error by $5.7\times$ to $47\times$ relative to state-of-the-art anomaly-based baselines. Deployed in Suricata, IDSForge sustains a throughput of 7.8 Gbps with a median per-packet latency overhead of 3.1 \textmu s. 
\end{abstract}


\maketitle
\pagestyle{plain}

\section{Introduction}
Web services such as online banking, e-commerce, and cloud collaboration are frequent targets of cyberattacks. Network intrusion detection systems (NIDS) deployed at the gateways of these services are an important line of defense~\cite{liao2013intrusion,shi2025helios}, but in real deployments, producing detection results alone is not enough. First, the NIDS inspects every flow to the service, so line-rate processing of Web-scale traffic with little overhead is essential~\cite{zhu2026proteus}. Second, benign traffic keeps changing as the service adds new pages and APIs, while attackers keep developing new techniques, which calls for detecting previously unseen attacks without rejecting new benign behavior~\cite{fu2023detecting}. Third, an alert may lead to costly actions such as blocking a client, so operators need to know \emph{why} a flow is flagged before they respond~\cite{jacobs2022ai,wei2023xnids}. Existing approaches address only some of these needs.

Signature-based NIDS such as Suricata~\cite{suricata2025} match traffic against human-readable rules efficiently, but they rely on manually written attack signatures and therefore miss unseen attacks. Supervised detectors~\cite{shapira2019flowpic,lin2022bert,zhou2025trafficformer} achieve high accuracy on known attacks, but they are trained on a fixed set of attack classes and often generalize poorly to new ones~\cite{fu2022frequency}. In contrast, anomaly-based detectors~\cite{garcia2009anomaly,mirsky2018kitsune,lian2025facing,han2024contramtd} learn only from benign traffic and report deviations as anomalies. Since they need no attack labels and can detect unseen attacks, they are well suited to changing Web environments.

However, anomaly-based detectors are still difficult to trust in practice. Deep anomaly detectors model benign traffic with neural networks, such as autoencoders~\cite{mirsky2018kitsune} or evidential models~\cite{lian2025facing}, and their decision boundaries are defined only implicitly by the learned representations. These boundaries may become outdated as traffic drifts (Figure~\ref{fig:intro1_a}) and may be sensitive to small perturbations (Figure~\ref{fig:intro1_b}), both of which cause many false positives. These detectors are also hard to interpret. Most explanation methods are designed for supervised detectors and explain a trained model \emph{post hoc}~\cite{guo2018lemna,jacobs2022ai,wei2023xnids}. The few methods for anomaly-based detectors either attribute individual alerts to input features~\cite{han2021deepaid} or extract surrogate rules~\cite{li2023interpreting,zhang2024dissect,han2024rules,li2024genos}. Such explanations only approximate the underlying model, and extracting surrogate rules requires costly exploration of the black-box decision boundary. As a result, the explanation given to operators may not match the logic of the deployed detector.

In this paper, we present IDSForge, a self-explaining anomaly-based NIDS. Like existing anomaly-based detectors, IDSForge learns only from benign traffic. Unlike them, its model outputs an \emph{explicit boundary} of benign traffic rather than an anomaly score. During training, the model maps benign flows to hyper-rectangles, and the union of these hyper-rectangles forms the benign region. The hyper-rectangles are then compiled into rules for popular NIDS engines such as Suricata. At runtime, the engine only matches flows against these rules and requires no neural inference. To find a suitable boundary form, we compare six learning paradigms and choose axis-aligned boundaries, as used by tree models, because they describe benign traffic accurately and can be directly expressed as rules. Each hyper-rectangle is a conjunction of range conditions (e.g., \texttt{120} $\le$ \texttt{pkt\_len\_mean} $\le$ \texttt{860}) that operators can read and NIDS engines can execute, and the conditions that a flagged flow violates directly explain the alert.
Realizing this idea involves three challenges.
(\rmnum{1}) \emph{How can the rules describe complex benign behavior while remaining compact enough to deploy?} Fitting a simple tree to a trained black-box detector limits expressiveness to the tree, whereas using many fine-grained regions makes the rule base too large.
(\rmnum{2}) \emph{How can the rules accept unseen benign variants without covering attacks?} The training data contain only part of the benign traffic, since users may visit Web pages that do not appear in it. Regions that are too narrow reject such variants, whereas regions that are too broad admit attacks.
(\rmnum{3}) \emph{How can the rules be updated with newly confirmed benign flows without losing previously learned behavior?} Editing rules directly may introduce conflicts and drop learned behavior, whereas retraining from scratch requires storing all historical traffic.

\begin{figure}[t]
  \centering
  \refstepcounter{figure}\setcounter{subfigure}{0}
  \newcommand{\intropanel}[3]{%
    \begin{minipage}[t]{0.32\linewidth}\centering
      \includegraphics[width=\linewidth]{#1}\par\vspace{1pt}
      {\scriptsize\refstepcounter{subfigure}\label{#2}(\alph{subfigure}) #3\par}
    \end{minipage}}
  \intropanel{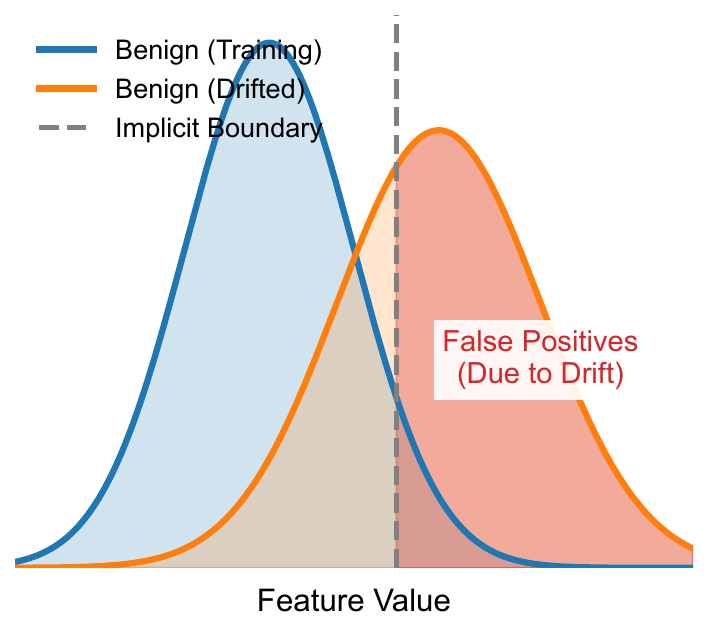}{fig:intro1_a}{Implicit boundaries become outdated under drift.}\hfill
  \intropanel{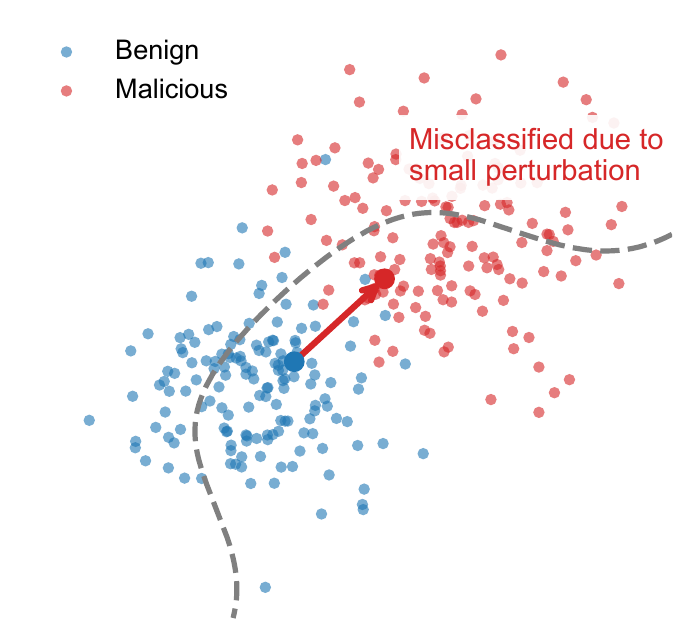}{fig:intro1_b}{Small perturbations cause misclassification.}\hfill
  \intropanel{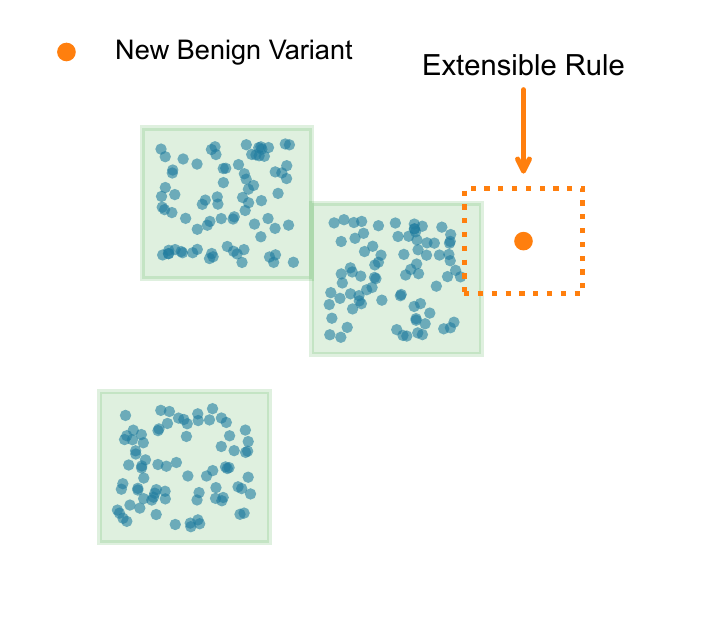}{fig:intro1_c}{IDSForge learns explicit, extensible rules.}
  \addtocounter{figure}{-1}
  \vspace{-0.2cm}
  \caption{Implicit boundaries of existing detectors (dashed lines) versus the explicit rules of IDSForge (green rectangles). Blue and red points are benign and malicious flows.}
  \vspace{-0.3cm}
\end{figure}

IDSForge addresses these challenges with three modules.
\textbf{Explicit distributional autoencoder (\S\ref{sec:bae}).} Conventional autoencoders detect anomalies by reconstruction error. Our explicit distributional autoencoder (EDAE) instead maps each benign flow to a hyper-rectangle that encloses its benign neighborhood, and the union of these hyper-rectangles forms the benign region. For example, an HTTPS flow with a mean packet length of 512 bytes and 14 forward packets is mapped to $430\le$ \texttt{pkt\_len\_mean} $\le 610 \wedge 9\le$ \texttt{fwd\_pkts} $\le 22$. To generalize to unseen benign traffic, EDAE reconstructs each flow from latent variables sampled uniformly within its hyper-rectangle, so the decoder has to recover the flow from any point inside. Additional objectives balance the compactness, generalization, and stability of the hyper-rectangles.
\textbf{High-speed rule compiler (\S\ref{sec:bre}).} The compiler converts the hyper-rectangles into native NIDS rules using templates. To support near-line-rate detection, it reduces the number of rules through Boolean minimization without changing the benign region, and indexes the rules with an R*-tree~\cite{beckmann1990r} for sub-linear matching.
\textbf{Self-distillation co-evolver (\S\ref{sec:blu}).} When the false-positive rate on operator-verified traffic exceeds a threshold, the co-evolver converts the rules of the affected server back into hyper-rectangles and samples benign flows from the latent space of EDAE within them. It then fine-tunes EDAE on these samples and the false positives and regenerates the rules. In this way, the model and the rules are updated together without storing the original training traffic or accumulating conflicting rules.

Our technical contributions are as follows.
\begin{itemize}
  \item We propose IDSForge, a self-explaining anomaly-based NIDS that trains a model to output explicit benign boundaries instead of anomaly scores, so that detection and explanation are based on the same executable rules.
  \item We design EDAE, an autoencoder that learns explicit bounds for each feature, which cover unseen benign traffic while keeping false positives low.
  \item We develop a high-speed rule compiler that converts the learned boundaries into native NIDS rules and indexes them for sub-linear matching, enabling high-throughput detection.
  \item We introduce a self-distillation co-evolver that refines EDAE with false positives and benign flows sampled within existing rules, so that the model and the rules adapt to traffic drift together.
  \item We prove that rule reduction is lossless, that the explanations are faithful and provide certified decision margins, and that the co-evolver does not forget previously learned rules (Appendix~\ref{sec:theory}).
\end{itemize}

We evaluate IDSForge on six public datasets and compare it with two supervised detectors, seven anomaly-based detectors, and four post-hoc explainers. We also deploy the learned rules in Suricata for online detection on traffic replayed at up to 10 Gbps. Our main findings are as follows.
\begin{itemize}
  \item \textbf{IDSForge accurately detects unseen attacks without attack labels.} It achieves F1 scores of at least 98.79\% on all six datasets, reduces the F1 error of the best anomaly-based baseline by $5.7\times$ to $47\times$, and performs on par with supervised detectors.
  \item \textbf{The explanations of IDSForge are faithful by design.} Explanations come from the detection rules themselves, whereas even rule extractors designed for security applications disagree with the detector they explain on up to 7.17\% of flows. Each rule is a short conjunction of conditions that operators can audit, edit, or remove.
  \item \textbf{IDSForge runs at high speed with low overhead.} In Suricata, it sustains 7.8 Gbps with a median per-packet latency overhead of 3.1~\textmu s, 0.31\% additional CPU utilization, and 42 MB of memory. Compiling the 15,850 rules of CIC-IDS2017 takes 0.8 seconds, and the online F1 scores are within 2.11 percentage points of the offline results.
  \item \textbf{IDSForge remains effective over long-term deployment.} Trained on one day of benign traffic, it achieves an F1 score of 97.55\% on a week of attacks in CIC-IDS2017, which rises to 99.65\% as more days are added. In comparison, the two strongest supervised baselines reach only 78.57\% and 63.29\% after two days.
\end{itemize}

Overall, IDSForge balances detection accuracy, interpretability, and efficiency. It provides end-to-end interpretable detection, in which the same explicit rules are used both to detect attacks and to explain alerts, making IDSForge practical for real-world deployment.

\section{Background and Motivation}
\label{sec:background}

\subsection{Problem Statement}
\label{sec:problem}
We consider an NIDS deployed at the gateway of a protected network, e.g., an enterprise perimeter or a cloud-hosted Web service. Each flow is a min-max normalized feature vector $x\in[0,1]^n$. The training set $\mathcal{D}$ contains only benign traffic without attack labels, and operators can confirm false positives during operation. The goal is to learn a detector $f:[0,1]^n\to\{0,1\}$ that flags unseen attacks as anomalies ($f(x)=1$).

\subsection{Anomaly-Based Network Intrusion Detection}
\label{sec:bg-detection}
Anomaly-based NIDS model benign traffic alone and flag deviations as attacks, so they can detect unseen attacks without attack labels~\cite{garcia2009anomaly,fu2022frequency}. Existing detectors characterize benign behavior by reconstruction or density~\cite{mirsky2018kitsune,marteau2021random,ruff2018deep}, frequency-domain patterns~\cite{fu2021realtime}, interaction structure~\cite{fu2023detecting,han2024contramtd}, or calibrated confidence~\cite{zhao2024trident,lian2025facing}. These detectors compare a scalar score with a threshold, so the benign region is only \emph{implicitly} defined by an opaque function. Such implicit boundaries are hard to interpret, become stale under traffic drift, and are sensitive to minor perturbations, leading to high false-positive rates.

\subsection{Limitations of Post-Hoc Explanation}
\label{sec:bg-interp}
Existing work explains a trained detector post hoc, either by attributing alerts to input features~\cite{guo2018lemna,wei2023xnids,han2021deepaid,han2023anomaly,yang2021cade} or by distilling the detector into rules or trees~\cite{jacobs2022ai,li2023interpreting,zhang2024dissect,han2024rules,li2024genos}. These approaches have two limitations. First, most of them, including rule-based in-network NIDS, require a supervised detector trained on attack labels~\cite{guo2018lemna,wei2023xnids,jacobs2022ai,zhou2023netbeacon,zhu2026proteus,shi2025helios}, and thus generalize poorly to unseen attacks. Even explainers for anomaly-based detectors~\cite{han2021deepaid,li2023interpreting,han2024rules} must first train and probe a separate detector. Second, a post-hoc explanation only approximates the detector after training and may disagree with it on some flows~\cite{li2023interpreting,jacobs2022ai,han2024rules}. Thus, the rule shown to operators may not reflect the alert logic~\cite{rudin2019stop}, and obtaining it requires costly boundary exploration or per-alert computation.

\subsection{Motivation}
\label{sec:motivation}
These limitations stem from separating detection from explanation. We therefore aim to learn bounded rules directly, i.e., a model that outputs an \emph{explicit boundary} of benign traffic, which serves as both the detector and its explanation. The form of this boundary is central to this goal. We compare six learning paradigms with different boundaries on CIC-IDS2017 (Table~\ref{table:motivationConstract} in Appendix~\ref{sec:appendix1}), among which tree models, whose boundaries are axis-aligned, achieve the highest F1 score (99.4\%) and fast inference. Axis-aligned boundaries have also proven effective for interpreting intrusion detectors~\cite{jacobs2022ai,li2023interpreting,han2024rules}, and NIDS engines natively support them as range conditions~\cite{suricata2025}. We thus represent the benign boundary as a union of tree-like axis-aligned hyper-rectangles $\{H_k\}$, i.e., $f(x)=\mathbb{1}[x\notin\bigcup_k H_k]$. Each $H_k$ is a readable conjunction of range conditions (e.g., \texttt{120} $\le$ \texttt{pkt\_len\_mean} $\le$ \texttt{860}), and an alert is explained by its violated conditions.

Two further design choices make such boundaries learnable from benign traffic alone. First, we learn the rules separately for each destination server IP. Traffic toward the same server is relatively stable, as it is driven by the services the server provides~\cite{xu2005profiling,karagiannis2005blinc}, so a few hyper-rectangles can enclose its benign behaviors~\cite{sommer2010outside}. Second, we build on the autoencoder, which suits modeling the latent distribution of benign traffic and is widely used in anomaly-based NIDS~\cite{mirsky2018kitsune,zhao2024trident,lian2025facing}. Sampling from its latent space also generates synthetic benign flows, so the model and the rules can be co-updated without storing historical traffic.

\section{IDSForge Design}
\label{sec:design}

\subsection{Overview}
\label{sec:overview}

\begin{figure}[t]
  \centering
  \includegraphics[width=0.92\linewidth]{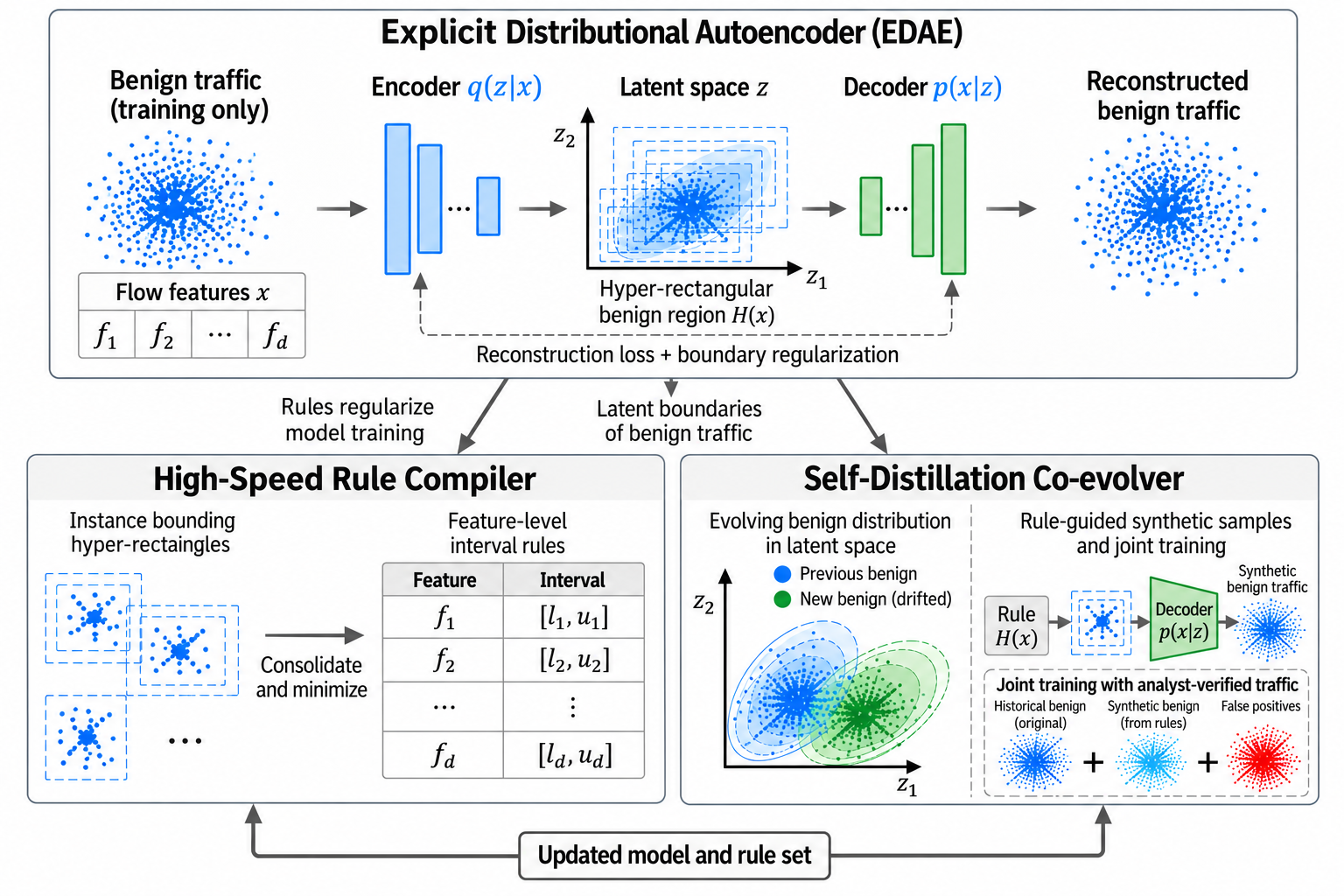}
  \vspace{-0.5cm}
  \caption{Workflow of IDSForge and its three modules.}
  \label{fig:sysOverview}
  \vspace{-0.4cm}
\end{figure}

To address the challenges of implicit decision boundaries and post-hoc explanation, IDSForge aims to learn explicit boundaries of benign traffic from benign flows alone and deploy them directly as NIDS rules, so that detection and explanation share the same logic. Since traffic toward the same server IP is relatively stable~\cite{xu2005profiling,karagiannis2005blinc}, IDSForge keeps a separate rule base for each destination server IP, which is built and maintained by three modules (Figure~\ref{fig:sysOverview}).
(1) \textbf{Explicit distributional autoencoder.} EDAE maps each benign flow to a hyper-rectangle in a latent space aligned with the input features, so each latent variable is the bound of a feature. It also predicts a mask that selects the features to constrain, and is trained with five complementary objectives that keep them tight, generalizable, and stable.
(2) \textbf{High-speed rule compiler.} The compiler reads them in a single forward pass, losslessly reduces the rules of each server by Boolean minimization, indexes them with an R*-tree for sub-linear matching, and translates them into native NIDS rules through templates. At runtime, a flow that matches no rule is flagged as malicious, and the violated conditions of its nearest rule explain the alert.
(3) \textbf{Self-distillation co-evolver.} The co-evolver collects operator-confirmed false positives. When their rate exceeds a threshold, it generates benign flows by sampling the latent space within the existing rules, fine-tunes EDAE on them together with the false positives, and regenerates the rules, so that the model and the rules co-evolve without forgetting previously accepted benign traffic.
Appendix~\ref{sec:theory} analyzes the three modules theoretically.

\subsection{Explicit Distributional Autoencoder}
\label{sec:bae}
This module learns the explicit boundary of benign traffic. For each benign flow, EDAE outputs a hyper-rectangle bounding its benign neighborhood in the original feature space, which later becomes a detection rule. EDAE (Figure~\ref{fig:BAE}) builds on the autoencoder, a simple and effective backbone for anomaly-based traffic detection~\cite{mirsky2018kitsune,zhao2024trident,lian2025facing}, but outputs hyper-rectangular explicit distribution boundaries.

\begin{figure}[t]
  \centering
  \includegraphics[width=\linewidth]{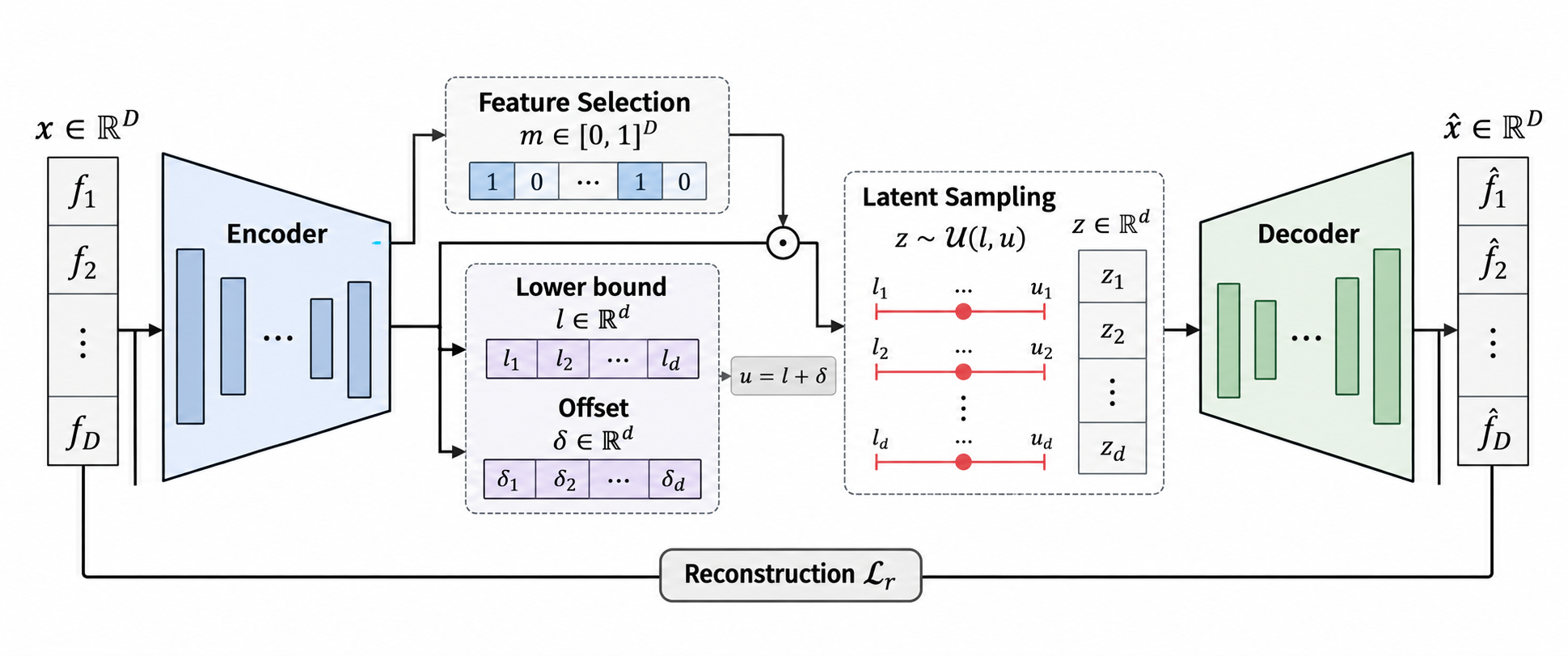}
  \vspace{-0.5cm}
  \caption{Architecture of EDAE, which maps a flow to per-feature bounds and decodes samples within them.}
  \label{fig:BAE}
  \vspace{-0.4cm}
\end{figure}

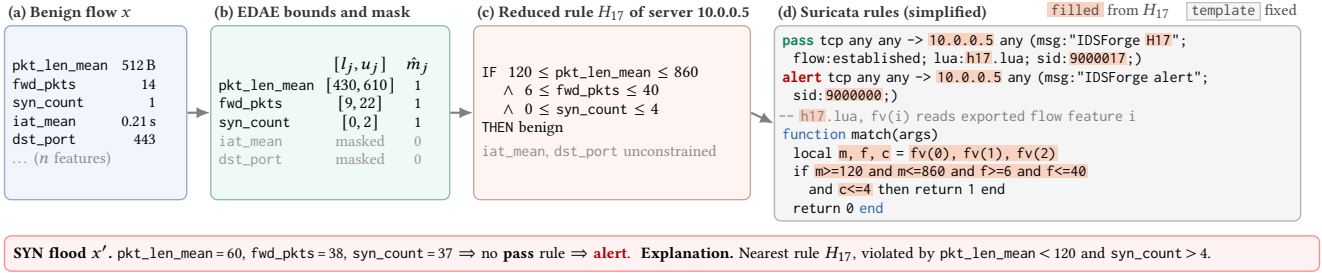
\begin{figure*}[t]
  \centering
  \providecommand{\dyn}[1]{{\setlength{\fboxsep}{0.4pt}\colorbox{cAlt!30}{#1}}}%
  \begin{tikzpicture}[
      font=\scriptsize,
      panel/.style={draw=#1!55!black!45, fill=#1!7, rounded corners=2.5pt, line width=0.5pt,
                    inner sep=3pt, minimum height=2.3cm, anchor=north west, align=left},
      hdr/.style={font=\scriptsize\bfseries, text=black!80, anchor=south west, inner sep=1.5pt},
      arr/.style={-{Latex[length=1.8mm,width=1.5mm]}, line width=0.8pt, black!50},
      kw/.style={},
    ]
    \node[panel=cPal, text width=2.2cm] (a) at (0,0) {%
      \begin{tabular}{@{}l@{\hspace{5pt}}r@{}}
        \texttt{pkt\_len\_mean} & 512\,B\\
        \texttt{fwd\_pkts} & 14\\
        \texttt{syn\_count} & 1\\
        \texttt{iat\_mean} & 0.21\,s\\
        \texttt{dst\_port} & 443\\
        \multicolumn{2}{@{}l@{}}{\textcolor{black!45}{\ldots\ ($n$ features)}}
      \end{tabular}};
    \node[hdr] at (a.north west) {(a) Benign flow $x$};
    \node[panel=cThr, text width=2.95cm, right=3mm of a.north east, anchor=north west] (b) {%
      \begin{tabular}{@{}l@{\hspace{4pt}}c@{\hspace{4pt}}c@{}}
        & $[l_j,u_j]$ & $\hat m_j$\\[1pt]
        \texttt{pkt\_len\_mean} & $[430,610]$ & 1\\
        \texttt{fwd\_pkts} & $[9,22]$ & 1\\
        \texttt{syn\_count} & $[0,2]$ & 1\\
        \textcolor{black!40}{\texttt{iat\_mean}} & \textcolor{black!40}{masked} & \textcolor{black!40}{0}\\
        \textcolor{black!40}{\texttt{dst\_port}} & \textcolor{black!40}{masked} & \textcolor{black!40}{0}
      \end{tabular}};
    \node[hdr] at (b.north west) {(b) EDAE bounds and mask};
    \node[panel=cAlt, text width=3.45cm, right=3mm of b.north east, anchor=north west] (c) {%
      \begin{tabular}{@{}l@{}}
        \texttt{IF}\ \ $120\le$ \texttt{pkt\_len\_mean} $\le 860$\\
        \quad\ $\wedge\ \ 6\le$ \texttt{fwd\_pkts} $\le 40$\\
        \quad\ $\wedge\ \ 0\le$ \texttt{syn\_count} $\le 4$\\
        \texttt{THEN} benign\\[2pt]
        \textcolor{black!45}{\texttt{iat\_mean}, \texttt{dst\_port} unconstrained}
      \end{tabular}};
    \node[hdr] at (c.north west) {(c) Reduced rule $H_{17}$ of server 10.0.0.5};
    \node[panel=black, fill=black!4, draw=black!30, text width=6.75cm, right=3mm of c.north east, anchor=north west] (d) {%
      \fontsize{6}{6.9}\selectfont\ttfamily\setlength{\fboxsep}{0.4pt}%
      \textcolor{cThr!75!black}{\textbf{pass}} tcp any any -> \dyn{10.0.0.5} any (msg:"IDSForge \dyn{H17}";\\
      \ \ flow:established; lua:\dyn{h17}.lua; sid:\dyn{9000017};)\\
      \textcolor{red!70!black}{\textbf{alert}} tcp any any -> \dyn{10.0.0.5} any (msg:"IDSForge alert";\\
      \ \ sid:\dyn{9000000};)\\[1pt]
      \textcolor{black!50}{-{}- \dyn{h17}.lua, fv(i) reads exported flow feature i}\\
      \textcolor{cPal!80!black}{function} match(args)\\
      \ \ local \dyn{m, f, c} = \dyn{fv(0), fv(1), fv(2)}\\
      \ \ if \dyn{m>=120 and m<=860 and f>=6 and f<=40}\\
      \ \ \ \ \ and \dyn{c<=4} then return 1 end\\
      \ \ return 0 \textcolor{cPal!80!black}{end}};
    \node[hdr] at (d.north west) {(d) Suricata rules (simplified)};
    \node[anchor=south east, inner sep=1.5pt, font=\scriptsize, text=black!70] at (d.north east)
      {\setlength{\fboxsep}{0.6pt}\colorbox{cAlt!30}{\ttfamily\,filled\,} from $H_{17}$\quad\fcolorbox{black!30}{black!4}{\ttfamily\,template\,} fixed};
    \draw[arr] (a.east) -- (b.west);
    \draw[arr] (b.east) -- (c.west);
    \draw[arr] (c.east) -- (d.west);
    \node[panel=red, fill=red!5, draw=red!45, minimum height=0pt, inner sep=3.5pt, text width=17.15cm,
          anchor=north west] (e) at ([yshift=-1.8mm]a.west |- d.south) {%
      \textbf{SYN flood $x'$.}\ \texttt{pkt\_len\_mean}\,=\,60, \texttt{fwd\_pkts}\,=\,38, \texttt{syn\_count}\,=\,37
      $\Rightarrow$ no \textbf{pass} rule $\Rightarrow$ \textcolor{red!70!black}{\textbf{alert}}.\ \
      \textbf{Explanation.}\ Nearest rule $H_{17}$, violated by \texttt{pkt\_len\_mean}\,$<$\,120 and \texttt{syn\_count}\,$>$\,4.};
  \end{tikzpicture}
  \vspace{-0.2cm}
  \caption{Rule compilation for a benign HTTPS flow. (a) Flow features. (b) EDAE bounds and mask. (c) Reduced rule $H_{17}$. (d) Suricata rules with fields filled from $H_{17}$ highlighted. The bottom bar explains an alert for a SYN flood flow.}
  \label{fig:ruleExample}
  \vspace{-0.3cm}
\end{figure*}

\noindent\textbf{Isomorphic interval space.}
Conventional autoencoders compress the input into a low-dimensional code~\cite{mienye2025deep} whose dimensions lack feature meaning. EDAE instead uses a latent space isomorphic to the input, where the $j$-th dimension encodes the bounds $[l_j,u_j]$ of the $j$-th feature. The encoder outputs a lower bound $l(x)$ and obtains the upper bound $u(x)$ by adding a squared offset, which guarantees $u_j\ge l_j$. Each flow thus maps to an instance-specific hyper-rectangular explicit distribution boundary describing its benign neighborhood,
\begin{equation}
H(x)=\prod_{j}\big[l_j(x),\,u_j(x)\big].
\end{equation}

\noindent\textbf{Feature-axis attention.}
A lightweight attention module reweights the input features and adds the result back through a residual connection. The residual form lets the encoder predict bounds as input offsets and stabilizes training.

\noindent\textbf{Feature selection.}
To constrain only useful features, a two-layer network with a sigmoid output maps the refined input to a soft mask $m(x)\in[0,1]^n$ that gates the bounds element-wise. Masked features, such as bursty timing or static fields, shrink to near-zero widths in training and thus drop out of the volume-based objectives. The mask is binarized at 0.5 into $\hat m$, and only the active features $\mathcal{A}(x)$ with $\hat m_j=1$ appear in the rule, which keeps rules concise.

\noindent\textbf{Uniform in-rectangle sampling.}
Unlike the unbounded Gaussian latent of conventional variational autoencoders, EDAE samples the latent variable uniformly within the bounds by reparameterizing $\epsilon\sim\mathcal{N}(0,I)$ as
\begin{equation}
z=l+\Phi(\epsilon)\odot(u-l),
\end{equation}
where $\Phi$ is the standard normal CDF, so $z\sim\mathcal{U}(l,u)$ by the probability integral transform. The decoder reconstructs $\hat{x}=g(z)$, so it must recover the input from any point of $H(x)$.

\noindent\textbf{Optimization objectives.}
EDAE jointly optimizes five complementary objectives over a mini-batch $\mathcal{B}$, computed on the active features.

\noindent\emph{(1) Reconstruction loss.} It requires recovering the input from points throughout $H(x)$, keeping the bounds consistent with it,
\begin{equation}
\mathcal{L}_1=\mathbb{E}_{z\sim\mathcal{U}(l,u)}\,\|x-g(z)\|_2^2.
\end{equation}

\noindent\emph{(2) Distribution loss.} It encourages $H(x)$ to cover its benign neighborhood rather than collapse to a point. We impose a uniform prior $p(z)=\mathcal{U}([0,1]^n)$, which factorizes over features and favors no region. With $\tilde q(z|x)=\mathcal{U}(l(x),u(x))$ and $H(x)\subseteq[0,1]^n$, the KL divergence has the closed form
\begin{equation}
\mathcal{L}_2=\mathrm{KL}\big(\tilde q(z|x)\,\|\,p(z)\big)=-\sum_{j}\log\big(u_j(x)-l_j(x)\big),
\end{equation}
which discourages collapse onto single samples and helps generalize to unseen benign variants.

\noindent\emph{(3) Density loss.} It favors a tight $H(x)$ containing many benign flows per unit volume,
\begin{equation}
\mathcal{L}_3=\frac{1}{|\mathcal{B}|}\sum_{x\in\mathcal{B}}\frac{V(x)}{N_{\mathrm{in}}(x)+\varepsilon},
\end{equation}
where $V(x)=\prod_{j\in\mathcal{A}(x)}(u_j-l_j)$ is the volume of $H(x)$, $N_{\mathrm{in}}(x)$ counts the flows of $\mathcal{B}$ inside $H(x)$, and $\varepsilon>0$ avoids division by zero. Since counting is non-differentiable, we use the smooth indicator
\begin{equation}
N_{\mathrm{in}}(x)\approx\sum_{x'\in\mathcal{B}}\prod_{j\in\mathcal{A}(x)}\sigma\Big(\frac{x'_j-l_j}{\tau}\Big)\,\sigma\Big(\frac{u_j-x'_j}{\tau}\Big),
\end{equation}
with the sigmoid $\sigma$ and temperature $\tau$ ($\tau=0.05$ by default). The product approaches 1 inside $H(x)$ and 0 outside, and becomes exact as $\tau\to 0$ (Proposition~\ref{prop:density}). $\mathcal{L}_3$ shrinks $H(x)$ toward dense benign regions, while $\mathcal{L}_2$ prevents over-contraction.

\noindent\emph{(4) Interval loss.} It keeps the bounds numerically stable,
\begin{equation}
\mathcal{L}_4=\lambda_1\|u-l\|_1+\lambda_2\sum_j\sigma(u_j-x_j)+\lambda_3\sum_j\sigma(x_j-l_j),
\end{equation}
where the first term penalizes overly wide intervals and the other two grow as the bounds move away from the flow, keeping the interval close to it. The saturating sigmoid bounds their pull, so they do not dominate $\mathcal{L}_2$ and $\mathcal{L}_3$.

\noindent\emph{(5) Entropy loss.} Jointly optimizing many features may yield a loose $H(x)$ that shifts uniformly along correlated directions. We normalize the margins into distributions $\pi^{u}=(u-x)/\|u-x\|_1$ and $\pi^{l}=(x-l)/\|x-l\|_1$ and minimize
\begin{equation}
\mathcal{L}_5=-\big(\mathrm{H}(\pi^{u})+\mathrm{H}(\pi^{l})\big),
\end{equation}
i.e., maximize their Shannon entropy. This prevents the margins from concentrating on a few features and adapts its shape to the variance of each feature.

The overall objective $\mathcal{L}$ is a weighted sum of the five terms, whose weights balance the gradient magnitudes of the terms (Figures~\ref{fig:sens-single}, \ref{fig:sens-inter}, and \ref{fig:sens-sobol} in Appendix~\ref{sec:sensitivity}).

\subsection{High-Speed Rule Compiler}
\label{sec:bre}

This module turns the learned $H(x)$ into a compact rule base that NIDS engines execute natively. Unlike post-hoc extractors that explore a black box~\cite{li2023interpreting,li2024genos}, it reads them from EDAE in a single forward pass (Figure~\ref{fig:ruleExample}).

\noindent\textbf{Rule generation.}
Each benign flow $x$ yields the rule $H(x)=\prod_{j\in\mathcal{A}(x)}[l_j(x),u_j(x)]$, i.e., \texttt{IF} $\bigwedge_{j\in\mathcal{A}(x)} l_j(x)\le x_j\le u_j(x)$ \texttt{THEN} benign, whose bounds are mapped back to raw feature units by inverting the min-max normalization. Grouped by destination server IP, a server receiving many diverse benign flows is described by a large set of overlapping hyper-rectangles $\mathcal{H}_s$, whose union $R(\mathcal{H}_s)$ is its benign region.

\noindent\textbf{Rule acceleration.}
Matching thousands of overlapping rules is too slow for line rate, so we use two techniques.

\emph{(1) Lossless rule reduction.} We decompose $R(\mathcal{H}_s)$ into as few hyper-rectangles as possible without changing the benign region, i.e., orthogonal polygon decomposition, which is NP-hard beyond two dimensions~\cite{dielissen1991rectangular}. We thus cast it as two-level Boolean minimization~\cite{brayton1984logic,rudell1987multiple}. The bounds in $\mathcal{H}_s$ cut each feature $j$ at a sorted set of points, and each cut point $c$ with index $\kappa_j(c)$ is encoded by a threshold variable $b_{j,\kappa_j(c)}=\mathbb{1}[x_j\ge c]$. The condition on feature $j$ then becomes
\begin{equation}
l\le x_j<u\ \iff\ b_{j,\kappa_j(l)}\wedge\neg\, b_{j,\kappa_j(u)},
\end{equation}
so each hyper-rectangle is a product term and $R(\mathcal{H}_s)$ is a sum of products $F_s$. Assignments that violate the order of the cut points never occur and form the don't-care set $\mathcal{C}$. We seek the fewest product terms $F'$ that cover $F_s$ without exceeding $F_s\vee\mathcal{C}$, which Espresso~\cite{brayton1984logic} solves heuristically by iterating three operations until $|F'|$ stops decreasing. Geometrically, \textsc{Expand} grows a hyper-rectangle until it hits the OFF-set $\neg(F_s\vee\mathcal{C})$, \textsc{Irredundant} drops those covered by others, and \textsc{Reduce} shrinks them to escape local minima. Since a conjunction of threshold literals on a feature is an interval, each product term decodes to a hyper-rectangle, so $\mathcal{H}'_s$ covers the same benign region with far fewer rules (Proposition~\ref{prop:reduction}).

\emph{(2) R-tree indexing.} We index $\mathcal{H}'_s$ with an R*-tree~\cite{beckmann1990r}, whose overlap- and margin-minimizing nodes suit point queries over overlapping rules. Each node $N$ stores the bounding box $B(N)$ of the rules beneath it. For high node utilization, the tree is bulk-loaded by Sort-Tile-Recursive packing~\cite{leutenegger1997str} into leaves of capacity $M$, and a rule added by the co-evolver descends into the child whose overlap with its siblings grows least. A point query recurses only into children whose bounding boxes contain $x$, finding matching rules in sub-linear time, and a flow is benign if at least one rule matches. A best-first search over $d_\infty(x,B(N))$, which lower-bounds the distance to the rules under $N$, finds the nearest rule of a flagged flow. Algorithm~\ref{alg:bre} in Appendix~\ref{sec:alg} summarizes the compilation.

\noindent\textbf{Rule translation.}
For an NIDS engine $E$, a template $\tau_E$ turns each reduced rule $H\in\mathcal{H}'_s$ into a \texttt{pass} rule $\tau_E(H,s)$ that matches flows to server $s$ satisfying all conditions of $H$, written as rule options or Lua functions in Suricata. A default rule $\alpha_E(s)$ raises an \texttt{alert} for every flow to $s$. Since pass rules precede alert rules, a flow is flagged exactly when it matches none of the rules of its server. An alert is explained by the nearest rule
\begin{equation}
H^*=\arg\min_{H\in\mathcal{H}'_s}d_\infty(x,H)
\end{equation}
and its violated conditions, where $d_\infty(x,H^*)$ bounds the feature change needed for evasion (Proposition~\ref{prop:margin}). As plain-text conjunctions with unique identifiers, rules can be audited, tightened, relaxed, disabled, or complemented by expert rules, taking effect after a reload without retraining.

\subsection{Self-Distillation Co-Evolver}
\label{sec:blu}
This module updates the model and the rules under traffic drift at low cost. Drift from new Web services and user behaviors raises false positives. Retraining from scratch must store all historical traffic, which is costly and raises privacy concerns, while patching rules directly may accumulate overlapping or contradictory rules over time and break their consistency with the model. The co-evolver thus regenerates rules through EDAE. Operators confirm false positives during alert triage, and the co-evolver buffers those of server $s$ in $\mathcal{F}_s$. An update is triggered only when the false-positive rate over a sliding window exceeds $\theta_{\mathrm{FP}}$, which avoids frequent retraining caused by sporadic false alarms.

\noindent\textbf{Self-distillation.}
Through self-distillation, the co-evolver replaces stored traffic with flows generated by EDAE itself~\cite{shin2017continual}. It inverts each rule of $\mathcal{H}'_s$ into a normalized hyper-rectangle $[l_H,u_H]$ by min-max normalizing its bounds, with inactive features at $[0,1]$, samples latent variables within each as in \S\ref{sec:bae}, and decodes them into synthetic benign flows,
\begin{equation}
\mathcal{G}_s=\{g(z)\mid z\sim\mathcal{U}(l_H,u_H),\ H\in\mathcal{H}'_s\}.
\end{equation}
As $\mathcal{L}_1$ trains the decoder to reconstruct benign flows from points throughout their rules, $\mathcal{G}_s$ approximates the benign flows from which the rules were learned, rather than arbitrary points in the rules. The co-evolver thus needs neither the original training traffic nor per-flow history.

\noindent\textbf{Model and rule co-evolution.}
EDAE is fine-tuned with $\mathcal{L}$ on $\mathcal{G}_s\cup\mathcal{F}_s$. The generated flows anchor the rules of previously learned behaviors, and the false positives extend them to the new benign behaviors, so the model adapts without catastrophic forgetting~\cite{kirkpatrick2017overcoming}. The refined EDAE then regenerates the rules of server $s$ from $\mathcal{G}_s\cup\mathcal{F}_s$, which are merged with $\mathcal{H}'_s$, reduced by Boolean minimization, and re-indexed in the R*-tree. Since the rules are produced by the model rather than patched by hand, the rule base stays consistent with EDAE and avoids long-term rule conflicts. As the merged union is reduced losslessly, accepted flows stay accepted (Proposition~\ref{prop:forgetting}). The update only affects server $s$, so the NIDS reloads only its rules, and other servers remain untouched.

\addtocounter{table}{1}
\begin{table*}[!t]
  \centering
  \caption{Binary detection results (\%). Green cells mark the best result in each column, and the shaded row is IDSForge.} 
  \label{tab:BinCsy}
  \tabstyle
  \begin{tabularx}{0.9\textwidth}{l *{9}{>{\centering\arraybackslash}X}}
    \tabtop
    \rowcolor{tabHead} \textbf{Model} & \multicolumn{3}{c}{\textbf{CIC-IDS2017}} & \multicolumn{3}{c}{\textbf{CSE-CIC-IDS2018}} & \multicolumn{3}{c}{\textbf{CIC-UNSW-NB15}} \\
    \cmidrule(lr){2-4} \cmidrule(lr){5-7} \cmidrule(lr){8-10}
    \rowcolor{tabHead} & ACC & AUC & F1 & ACC & AUC & F1 & ACC & AUC & F1 \\
    \tabmid
    ET-BERT     & 99.63 & \best{99.87} & 99.63 & 97.46 & 98.12 & 97.49 & 95.02 & 95.84 & 94.89 \\
    TMG-GAN      & 98.30 & 98.32 & 98.30 & 98.15 & 96.47 & 96.34 & 98.04 & 98.45 & 98.04 \\
    \tabmid
    Whisper     & 90.06 & 91.24 & 90.49 & 89.57 & 90.68 & 90.04 & 91.49 & 92.35 & 91.81 \\
    HyperVision & 91.17 & 92.05 & 91.51 & 90.91 & 91.84 & 91.27 & 90.42 & 91.36 & 90.82 \\
    Kitsune     & 57.13 & 56.19 & 54.99 & 82.76 & 83.27 & 82.42 & 76.73 & 77.69 & 76.11 \\
    Trident     & 77.42 & 82.99 & 75.17 & 80.34 & 82.15 & 81.96 & 78.85 & 81.12 & 80.72 \\
    UnDiff      & 83.31 & 88.88 & 83.72 & 82.31 & 84.76 & 83.63 & 81.98 & 84.15 & 83.35 \\
    Diff-RF     & 90.03 & 90.85 & 90.01 & 92.38 & 93.14 & 92.26 & 88.96 & 89.72 & 88.85 \\
    ContraMTD   & 93.92 & 94.86 & 94.14 & 92.37 & 96.48 & 94.82 & 93.96 & 95.12 & 94.17 \\
    \tabmid
    \rowcolor{tabOurs} \textbf{IDSForge} & \best{99.65} & 99.60 & \best{99.65} & \best{99.89} & \best{99.89} & \best{99.89} & \best{98.98} & \best{99.24} & \best{98.97} \\
    \tabmid
    \rowcolor{tabHead} & \multicolumn{3}{c}{\textbf{TON-IoT}} & \multicolumn{3}{c}{\textbf{Kitsune}} & \multicolumn{3}{c}{\textbf{CICIoT2023}} \\
    \cmidrule(lr){2-4} \cmidrule(lr){5-7} \cmidrule(lr){8-10}
    \rowcolor{tabHead} & ACC & AUC & F1 & ACC & AUC & F1 & ACC & AUC & F1 \\
    \tabmid
    ET-BERT     & 98.42 & 98.75 & 98.36 & 97.85 & 98.12 & 97.69 & 96.84 & 97.31 & 96.72 \\
    TMG-GAN      & 98.24 & 98.10 & 98.24 & 96.74 & 96.15 & 96.08 & 97.18 & 96.85 & 96.94 \\
    \tabmid
    Whisper     & 89.34 & 90.12 & 89.65 & 88.45 & 89.76 & 88.91 & 88.72 & 89.54 & 88.39 \\
    HyperVision & 90.48 & 91.35 & 90.82 & 89.62 & 90.84 & 89.95 & 89.85 & 90.76 & 90.12 \\
    Kitsune     & 79.94 & 81.25 & 79.48 & 85.14 & 86.66 & 84.72 & 82.15 & 79.84 & 78.36 \\
    Trident     & 79.16 & 81.42 & 80.05 & 78.63 & 80.95 & 79.18 & 77.92 & 80.14 & 78.65 \\
    UnDiff      & 82.75 & 84.18 & 83.12 & 83.40 & 85.12 & 83.85 & 82.69 & 83.25 & 81.93 \\
    Diff-RF     & 89.72 & 90.45 & 89.58 & 87.35 & 87.92 & 87.16 & 88.46 & 89.15 & 88.24 \\
    ContraMTD   & 93.15 & 94.68 & 93.42 & 91.84 & 93.52 & 92.19 & 92.64 & 94.10 & 92.88 \\
    \tabmid
    \rowcolor{tabOurs} \textbf{IDSForge} & \best{99.48} & \best{99.62} & \best{99.48} & \best{99.12} & \best{99.38} & \best{99.15} & \best{98.76} & \best{99.14} & \best{98.79} \\
    \tabbot
  \end{tabularx}
  \vspace{-0.05cm}
\end{table*}

\section{Evaluation}
\label{sec:eval}
We evaluate IDSForge through five research questions. \textbf{RQ1.} How accurately does IDSForge detect malicious traffic compared with state-of-the-art methods? \textbf{RQ2.} Are the learned rules interpretable, and how do they compare with post-hoc explainers? \textbf{RQ3.} Can IDSForge detect attacks continuously as traffic accumulates? \textbf{RQ4.} How does each module and objective contribute to detection performance? \textbf{RQ5.} Can IDSForge run in popular NIDS engines at near line rate?

\subsection{Experimental Setup}
\noindent\textbf{Datasets.} We use six public NIDS datasets covering enterprise, IoT, and hybrid environments (Table~\ref{tab:datasets}). We split each dataset 9:1 into training and test sets. Supervised detectors are trained on all training flows, whereas anomaly-based detectors and IDSForge use only the benign ones, and all methods share one test set with equal numbers of benign and malicious flows. The continuous detection experiment instead splits traffic by day, and the deployment experiment replays the test sets.

\setcounter{table}{0}
\begin{table}[t]
  \centering
  \caption{Dataset statistics from the released CSV files. Dur. is the capture duration, $\sim$ marks approximate values, \#Classes counts attack classes, and ``--'' means not published.}
  \label{tab:datasets}
  \tabstyle
  \setlength{\tabcolsep}{1.6pt}
  \rowcolors{2}{}{black!3}
  \begin{tabular*}{\linewidth}{@{\hspace{\tabcolsep}\extracolsep{\fill}}llrrrrr}
    \tabtop
    \rowcolor{tabHead} \textbf{Dataset} & Environment & \#Feat. & Size & Dur. & \makecell[r]{\#Benign\,/\\\#Malicious} & \makecell[r]{\#Cla-\\sses} \\
    \tabmid
    CIC-IDS2017     & Enterprise, 25 users & 78  & 51\,GB  & 5\,d & 2.27M\,/\,0.56M & 7 \\
    CSE-CIC-IDS2018 & Cloud, 420 hosts     & 79  & $\sim$450\,GB & 8\,d & 13.48M\,/\,2.75M & 7 \\
    CIC-UNSW-NB15   & Hybrid (IXIA)        & 49  & 100\,GB & $\sim$2\,d & 0.36M\,/\,0.09M & 9 \\
    TON-IoT         & IoT/IIoT testbed     & 44  & $\sim$6.27\,GB & -- & 22.34M total & 9 \\
    Kitsune         & IoT (camera, router) & 115 & $\sim$1--2\,GB & 9.5\,h & 2.7M total & 9 \\
    CICIoT2023      & IoT, 105 devices     & 46  & 548\,GB & 5\,d & 1.10M\,/\,45.85M & 33 \\
    \tabbot
  \end{tabular*}
  \vspace{-0.05cm}
\end{table}
\setcounter{table}{2}

\noindent\textbf{Baselines.} For detection, we compare IDSForge with two supervised detectors (ET-BERT~\cite{lin2022bert} and TMG-GAN~\cite{ding2024tmggan}) and seven anomaly-based detectors (Whisper~\cite{fu2021realtime}, HyperVision~\cite{fu2023detecting}, Kitsune~\cite{mirsky2018kitsune}, Trident~\cite{zhao2024trident}, UnDiff~\cite{lian2025facing}, Diff-RF~\cite{marteau2021random}, and Contra\-MTD~\cite{han2024contramtd}). For interpretability, we compare it with the local explainers LIME~\cite{ribeiro2016why} and LEMNA~\cite{guo2018lemna} and the global surrogates Trustee~\cite{jacobs2022ai} and SCD-Tree~\cite{zhang2024dissect} (Table~\ref{tab:baselines} in Appendix~\ref{sec:baselines}).

\noindent\textbf{Metrics and implementation.} We report accuracy (ACC), F1 score (F1), and area under the ROC curve (AUC), for which IDSForge scores a flow by its distance $\eta(x)$ to the nearest rule. The ablation study also reports precision (PRE) and recall (REC). For interpretability, we follow~\cite{li2023interpreting,zhang2024dissect} and measure the fidelity (FD) of an explanation model $g$ to the detector $f$ it explains as the fraction of test flows, both benign and malicious, on which their binary decisions agree. FD thus compares $g$ with $f$ rather than with the ground truth, whereas the true positive rate (TPR) and true negative rate (TNR) of $g$ are computed against the ground-truth labels. Experiments run on a machine with an NVIDIA RTX 4090 GPU and 128 GB of RAM (Python 3.10, PyTorch 2.1), and the R*-tree uses libspatialindex (Appendix~\ref{sec:setup}).

\subsection{Detection Performance (RQ1)}
\label{sec:overall}
As anomaly-based methods separate benign from anomalous traffic, we evaluate binary detection. As shown in Table~\ref{tab:BinCsy}, IDSForge achieves the highest F1 score on CIC-IDS2017, CSE-CIC-IDS2018, and CIC-UNSW-NB15 (99.65\%, 99.89\%, and 98.97\%), outperforming the strongest anomaly-based baseline, ContraMTD, by 5.51, 5.07, and 4.80 percentage points, respectively. On TON-IoT, Kitsune, and CICIoT2023, IDSForge achieves the highest AUC and F1 scores, with F1 scores of 99.48\%, 99.15\%, and 98.79\%, outperforming ContraMTD by 6.06, 6.96, and 5.91 percentage points, respectively.
IDSForge also matches supervised detectors, since each hyper-rectangle covers a dense benign region instead of separating labeled classes and thus does not depend on which attacks are labeled. In contrast, Kitsune applies one global threshold to an opaque score, which cannot fit heterogeneous benign regions, and TMG-GAN may overfit the synthetic samples of rare attacks.

\takeaway{Describing heterogeneous benign traffic with many local explicit regions, rather than one global score threshold, lets a detector trained only on benign traffic match supervised ones.}

\subsection{Interpretability Analysis (RQ2)}
\label{sec:interp}
\begin{figure}[!t]
  \centering
  \subfigbottomskip=-4pt
  \subfigcapskip=-5pt
  \subfigure[Flows and learned rules.]{
    \label{fig:pca_a}
    \includegraphics[width=0.47\linewidth]{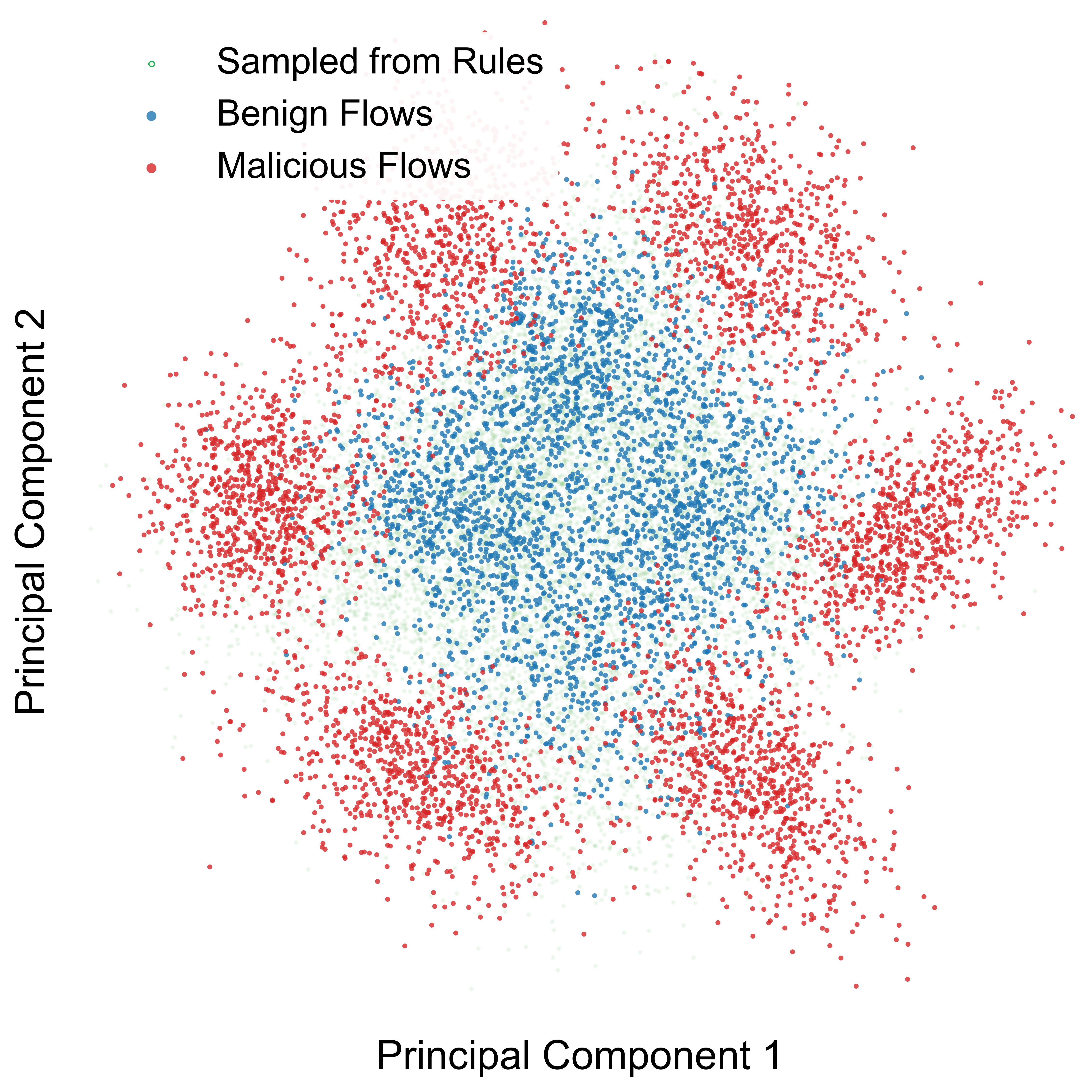}
  }
  \subfigure[Encoder representations (t-SNE).]{
    \label{fig:pca_b}
    \includegraphics[width=0.47\linewidth]{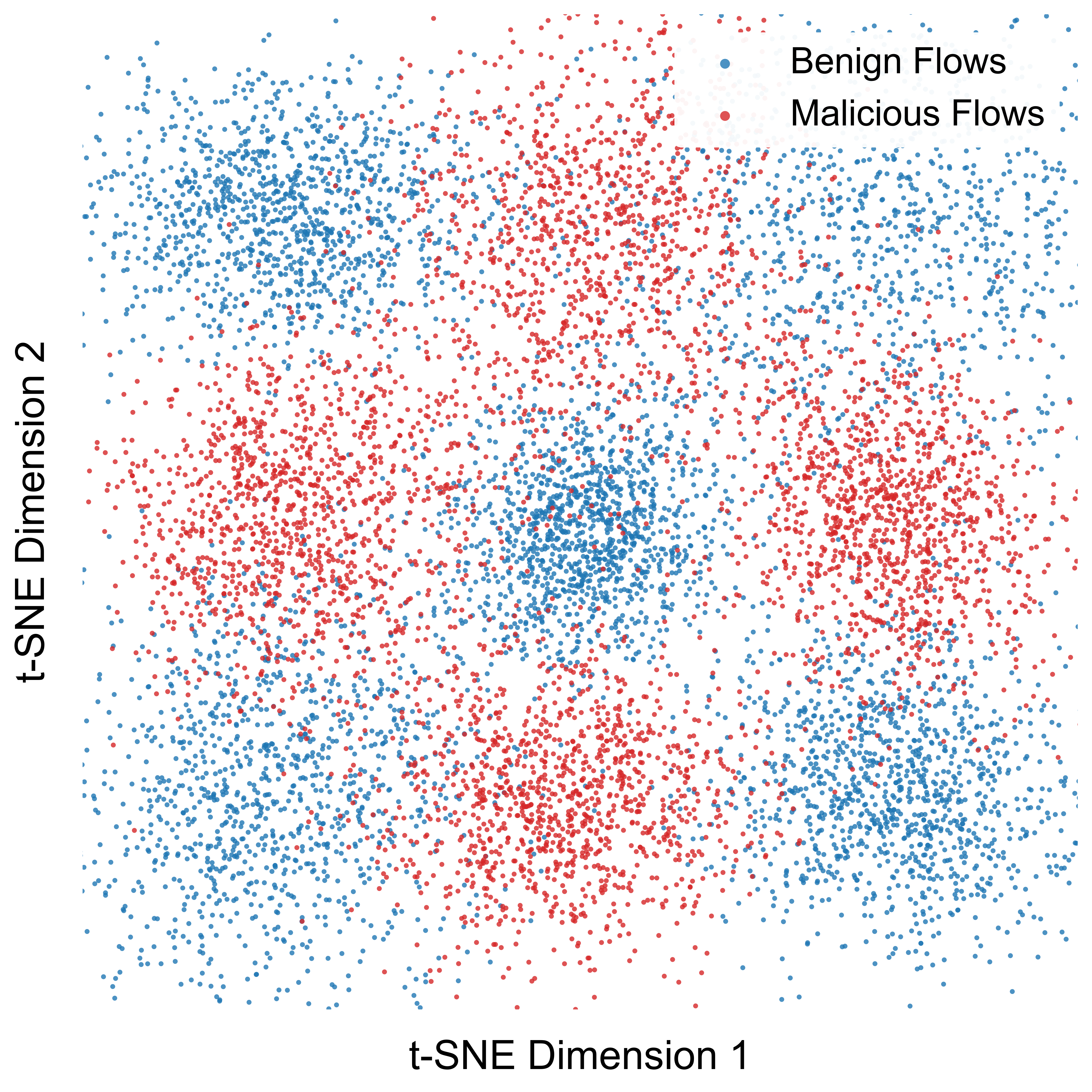}
  }
  \vspace{-0.2cm}
  \caption{Visualization of 20,000 CIC-IDS2017 flows. Blue and red points are benign and malicious flows, and green points in (a) form the rule-accepted region.}
  \label{fig:pca}
  \vspace{-0.4cm}
\end{figure}

\begin{table}[t]
  \centering
  \caption{Comparison with post-hoc explainers (\%) on the same autoencoder, with local explainers (top) aggregated by submodular pick. Orange cells mark values below 60\%, green cells mark full fidelity, and the shaded row is IDSForge.}
  \label{tab:posthoc}
  \tabstyle
  \setlength{\tabcolsep}{2.2pt}
  \begin{tabular}{lccccccccc}
    \tabtop
    \rowcolor{tabHead} & \multicolumn{3}{c}{\textbf{CIC-IDS2017}} & \multicolumn{3}{c}{\textbf{CSE-CIC-IDS2018}} & \multicolumn{3}{c}{\textbf{TON-IoT}} \\
    \cmidrule(lr){2-4} \cmidrule(lr){5-7} \cmidrule(lr){8-10}
    \rowcolor{tabHead} \textbf{Method} & FD & TPR & TNR & FD & TPR & TNR & FD & TPR & TNR \\
    \tabmid
    LIME~\cite{ribeiro2016why} & \weak{15.89} & 78.89 & \weak{0.14} & \weak{15.17} & 76.81 & \weak{2.55} & \weak{13.42} & 79.39 & \weak{0.27} \\
    LEMNA~\cite{guo2018lemna} & \weak{16.94} & 81.24 & \weak{0.86} & \weak{16.06} & 78.92 & \weak{3.18} & \weak{14.21} & 81.68 & \weak{0.92} \\
    \tabmid
    Trustee~\cite{jacobs2022ai} & \weak{50.71} & 99.81 & \weak{38.44} & \weak{59.56} & 98.57 & \weak{51.57} & \weak{46.85} & \weak{45.02} & 99.71 \\
    SCD-Tree~\cite{zhang2024dissect} & 94.75 & 99.37 & 93.60 & 93.01 & 98.92 & 91.80 & 92.83 & 92.50 & 99.15 \\
    \tabmid
    \rowcolor{tabOurs} \textbf{IDSForge} & \best{100.0} & 99.64 & 99.66 & \best{100.0} & 99.88 & 99.90 & \best{100.0} & 99.46 & 99.50 \\
    \tabbot
  \end{tabular}
  \vspace{-0.05cm}
\end{table}

To examine the learned rules, we visualize 20,000 flows from CIC-IDS2017. Figure~\ref{fig:pca_a} projects the flows and points sampled from the rule base onto the first two principal components. The green region accepted by the rules covers almost all benign flows, while most malicious flows lie in the gaps between rules, e.g., the clusters at the top, and fewer than 1\% fall inside it. The fragmented region reflects many instance-specific hyper-rectangular explicit distribution boundaries that match heterogeneous benign traffic. In the latent space, Figure~\ref{fig:pca_b} shows the t-SNE of the encoder representations. Benign flows form compact clusters and most malicious flows lie at the periphery, so EDAE learns the overall benign distribution from benign traffic alone. Since the projection pools the flows of all servers instead of viewing each server separately, some benign and malicious flows still mix near cluster boundaries, e.g., at the top left, which the per-server rules of Figure~\ref{fig:pca_a} separate.

\noindent\textbf{Comparison with post-hoc explainers.}
Following~\cite{li2023interpreting}, the black box is an autoencoder (AE) trained on benign flows that flags a flow when its reconstruction error exceeds a threshold. Each explainer accesses only benign flows and the AE decisions, and is turned into a surrogate $g$ that labels flows by itself. (\rmnum{1}) LIME and LEMNA output feature weights rather than detection results, so we aggregate their local models, a sparse linear model and a mixture regression fitted on AE-labeled perturbations around a flow, into a global surrogate by submodular pick~\cite{ribeiro2016why}. (\rmnum{2}) Trustee and (\rmnum{3}) SCD-Tree distill the AE decisions into a decision tree and Gaussian-bounded interval rules, respectively. We report the FD of $g$ to the AE on test flows with the original class ratio~\cite{li2023interpreting} and the TPR and TNR of $g$, which describe the surrogate rather than the AE. For IDSForge, $g$ is the deployed detector, so its FD attains the maximum by definition, and the meaningful comparison lies in TPR and TNR.

As shown in Table~\ref{tab:posthoc}, the local explainers yield imbalanced rules. LIME and LEMNA reach TPRs of 76.81\%--81.68\% but TNRs of at most 3.18\%, i.e., they reject almost all benign flows, and their fidelity is only 13.42\%--16.94\%. Trustee disagrees with the AE on 40.44\%--53.15\% of flows, likely because one decision tree approximates the nonlinear AE boundary only coarsely. SCD-Tree reaches a fidelity of 92.83\%--94.75\% but still disagrees on 5.25\%--7.17\% of flows, so the rule shown to operators may not reflect the alert logic. In contrast, IDSForge keeps its explanations faithful while achieving a TPR of at least 99.46\% and a TNR of at least 99.50\%. Only Trustee exceeds it on single metrics (a TPR of 99.81\% on CIC-IDS2017 and a TNR of 99.71\% on TON-IoT), at the cost of a TNR of 38.44\% and a TPR of 45.02\%, respectively.

\takeaway{Explanations are guaranteed to match alerts only when the explanation is the detector itself, since even high-fidelity post-hoc surrogates still disagree on some flows.}

\subsection{Continuous Detection (RQ3)}
\label{sec:drift}
We simulate long-term detection by splitting CIC-IDS2017 into its five days, where Monday contains only benign traffic. After each day, IDSForge is updated with that day's benign training flows, and ET-BERT and TMG-GAN are retrained from Tuesday on all labeled flows so far. All models are tested on the test set of \S\ref{sec:overall}, whose malicious flows span the whole week, simulating retrospective detection of future attacks. For fairness, IDSForge is updated at day boundaries rather than by the false-positive trigger.

\begin{figure}[t]
  \centering
  \begin{tikzpicture}
    \begin{groupplot}[
      paperstyle, linelegend,
      group style={group size=2 by 1, horizontal sep=1.0cm},
      width=0.56\linewidth, height=3.5cm,
      xmin=12, xmax=108, ymin=50, ymax=106, ytick={60,80,100},
      xtick={20,40,60,80,100}, xticklabels={Mon,+Tue,+Wed,+Thu,+Fri},
      xlabel={Days used for training},
      every axis plot/.append style={line width=1pt, mark size=2pt},
    ]
      \nextgroupplot[ylabel={F1 (\%)}, legend to name=driftlegend]
        \addplot[cPal, mark=square*] coordinates {(20,97.55) (40,98.06) (60,98.54) (80,99.07) (100,99.65)};
        \addlegendentry{IDSForge}
        \addplot[cThr, dashed, mark=triangle*, mark options={solid, fill=white}] coordinates {(40,78.57) (60,83.75) (80,93.53) (100,98.30)};
        \addlegendentry{TMG-GAN}
        \addplot[cAlt, densely dotted, mark=*, mark options={solid, fill=white}] coordinates {(40,63.29) (60,77.78) (80,87.64) (100,99.63)};
        \addlegendentry{ET-BERT}
        \draw[{Latex[length=1.3mm]}-{Latex[length=1.3mm]}, black!55, line width=0.5pt] (axis cs:44,80.4) -- (axis cs:44,96.3);
        \node[font=\scriptsize, text=black!75, anchor=west, inner sep=1pt] at (axis cs:45,88.4) {+19.5};
      \nextgroupplot[ylabel={AUC (\%)}]
        \addplot[cPal, mark=square*] coordinates {(20,97.56) (40,98.06) (60,98.56) (80,99.08) (100,99.60)};
        \addplot[cThr, dashed, mark=triangle*, mark options={solid, fill=white}] coordinates {(40,78.57) (60,83.62) (80,93.51) (100,98.32)};
        \addplot[cAlt, densely dotted, mark=*, mark options={solid, fill=white}] coordinates {(40,62.70) (60,77.54) (80,87.58) (100,99.87)};
        \draw[{Latex[length=1.3mm]}-{Latex[length=1.3mm]}, black!55, line width=0.5pt] (axis cs:44,80.4) -- (axis cs:44,96.3);
        \node[font=\scriptsize, text=black!75, anchor=west, inner sep=1pt] at (axis cs:45,88.4) {+19.5};
    \end{groupplot}
  \end{tikzpicture}

  \vspace{1pt}
  \pgfplotslegendfromname{driftlegend}
  \vspace{-0.2cm}
  \caption{Continuous detection on CIC-IDS2017. Models are updated daily and tested on the test set of Table~\ref{tab:BinCsy}, and arrows mark the gap after Tuesday.}
  \label{res:inc}
  \vspace{-0.3cm}
\end{figure}
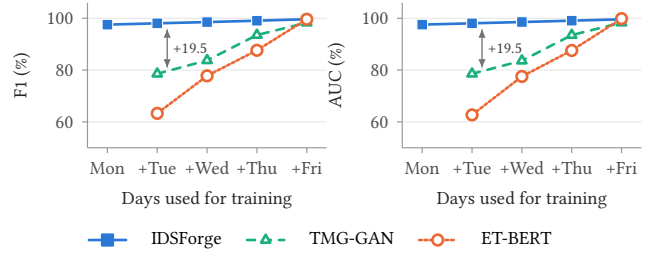

As shown in Figure~\ref{res:inc}, IDSForge reaches an F1 score of 97.55\% after Monday alone and detects attacks from later days, as $\mathcal{L}_3$ captures dense benign regions early. With each added day, the mean normalized rule volume grows from 0.12 to 0.68 and F1 rises monotonically to 99.65\%. In contrast, TMG-GAN and ET-BERT reach 78.57\% and 63.29\% after Tuesday and become comparable only after the full week. They need the complete data distribution and struggle with drift and long-term detection, whereas IDSForge performs well with partial data.

\takeaway{Since the benign traffic toward a server is largely known and stable while future attacks are not, modeling per-server benign behavior detects attacks absent from training, whereas supervised detectors only learn the attacks they have seen.}

\addtocounter{figure}{1}
\begin{figure}[t]
  \centering
  \begin{tikzpicture}
    \begin{axis}[
      paperstyle,
      width=\linewidth, height=3.7cm,
      xmin=-0.8, xmax=10.8, ymin=96, ymax=100.6, ytick={96,97,98,99,100},
      ylabel={Score (\%)},
      xtick={0,1,2,4,5,6,8,9,10},
      xticklabels={ACC,AUC,F1,ACC,AUC,F1,ACC,AUC,F1},
      x tick label style={font=\scriptsize, text=black!65},
      xtick style={draw=none}, clip=false,
      legend style={at={(0.5,1.03)}, anchor=south, legend columns=4, column sep=3pt,
                    /tikz/every even column/.append style={column sep=6pt}},
    ]
      \node[font=\footnotesize, text=black!85, anchor=north, yshift=-9pt] at (axis cs:1,96) {CIC-IDS2017};
      \node[font=\footnotesize, text=black!85, anchor=north, yshift=-9pt] at (axis cs:5,96) {CSE-CIC-IDS2018};
      \node[font=\footnotesize, text=black!85, anchor=north, yshift=-9pt] at (axis cs:9,96) {CIC-UNSW-NB15};
      \draw[cGrid, line width=0.4pt] (axis cs:3,96) -- (axis cs:3,100.6);
      \draw[cGrid, line width=0.4pt] (axis cs:7,96) -- (axis cs:7,100.6);
      \addplot[black!45, dash pattern=on 3pt off 1.5pt, line width=0.8pt, forget plot] coordinates {(-0.5,99.65) (2.5,99.65)};
      \addplot[black!45, dash pattern=on 3pt off 1.5pt, line width=0.8pt, forget plot] coordinates {(3.5,99.89) (6.5,99.89)};
      \addplot[black!45, dash pattern=on 3pt off 1.5pt, line width=0.8pt, forget plot] coordinates {(7.5,98.97) (10.5,98.97)};
      \addplot[only marks, mark=square*, mark size=2.1pt, mark options={fill=cAlt, draw=white, line width=0.4pt}]
        coordinates {(0,97.53) (4,98.59) (8,98.27)};
      \addlegendentry{Online ACC}
      \addplot[only marks, mark=triangle*, mark size=2.7pt, mark options={fill=cThr, draw=white, line width=0.4pt}]
        coordinates {(1,97.54) (5,98.59) (9,98.28)};
      \addlegendentry{Online AUC}
      \addplot[only marks, mark=*, mark size=2.3pt, mark options={fill=cPal, draw=white, line width=0.4pt}]
        coordinates {(2,97.54) (6,98.59) (10,98.28)};
      \addlegendentry{Online F1}
      \addlegendimage{black!45, dash pattern=on 3pt off 1.5pt, line width=0.8pt}
      \addlegendentry{Offline F1}
      \node[font=\scriptsize, text=black!80, anchor=south, inner sep=1pt, yshift=3pt] at (axis cs:0,97.53) {97.53};
      \node[font=\scriptsize, text=black!80, anchor=south, inner sep=1pt, yshift=3pt] at (axis cs:1,97.54) {97.54};
      \node[font=\scriptsize, text=black!80, anchor=south, inner sep=1pt, yshift=3pt] at (axis cs:2,97.54) {97.54};
      \node[font=\scriptsize, text=black!80, anchor=south, inner sep=1pt, yshift=3pt] at (axis cs:4,98.59) {98.59};
      \node[font=\scriptsize, text=black!80, anchor=south, inner sep=1pt, yshift=3pt] at (axis cs:5,98.59) {98.59};
      \node[font=\scriptsize, text=black!80, anchor=south, inner sep=1pt, yshift=3pt] at (axis cs:6,98.59) {98.59};
      \node[font=\scriptsize, text=black!80, anchor=south, inner sep=1pt, yshift=3pt] at (axis cs:8,98.27) {98.27};
      \node[font=\scriptsize, text=black!80, anchor=south, inner sep=1pt, yshift=3pt] at (axis cs:9,98.28) {98.28};
      \node[font=\scriptsize, text=black!80, anchor=south, inner sep=1pt, yshift=3pt] at (axis cs:10,98.28) {98.28};
    \end{axis}
  \end{tikzpicture}
  \vspace{-0.2cm}
  \caption{Online detection performance of IDSForge in Suricata. Markers give the online ACC, AUC, and F1, and dashed lines mark the offline F1 for reference.}
  \label{res:online}
  \vspace{-0.3cm}
\end{figure}
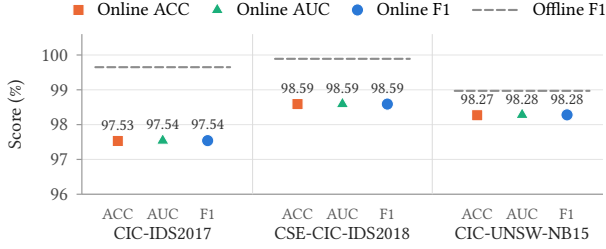
\addtocounter{figure}{-2}

\subsection{Ablation Study (RQ4)}
\label{sec:ablation}
We study the contribution of each module and objective on CIC-IDS2017 using four types of variants. \emph{w/o EDAE} replaces EDAE with a conventional VAE, \emph{w/o Compiler} detects anomalies using the reconstruction error of EDAE instead of rule matching, \emph{w/o Co-evolver} updates only the rules without refining EDAE, and \emph{w/o $\mathcal{L}_k$} removes the $k$-th objective.

\begin{figure}[t]
  \centering
  \begin{tikzpicture}
    \begin{groupplot}[
      paperstyle,
      group style={group size=2 by 1, horizontal sep=0.3cm, y descriptions at=edge left},
      height=3.3cm, ymin=0, ymax=108, ytick={0,25,50,75,100},
      ybar=0.6pt,
      enlarge x limits=false,
      x tick label style={font=\scriptsize, rotate=30, anchor=north east, inner sep=1pt},
      xtick style={draw=none},
      ylabel={Score (\%)},
      every axis plot/.append style={outbar}, barlegend,
    ]
      \nextgroupplot[width=0.575\linewidth, xmin=-0.6, xmax=3.6, /pgf/bar width=4pt, xtick={0,1,2,3}, xticklabels={Full,w/o EDAE,w/o Compiler,w/o Co-evolver}, title={(a) Modules}, legend to name=ablegend]
        \addplot[fill=cPal] coordinates {(0,99.66) (1,64.21) (2,91.24) (3,80.18)};
        \addlegendentry{PRE}
        \addplot[fill=cAlt] coordinates {(0,99.64) (1,66.83) (2,77.89) (3,93.45)};
        \addlegendentry{REC}
        \addplot[fill=cThr, postaction={pattern=north east lines, pattern color=white}] coordinates {(0,99.65) (1,65.49) (2,84.03) (3,86.31)};
        \addlegendentry{F1}
        \draw[black!55, line width=0.5pt, dash pattern=on 2.5pt off 1.5pt] (axis cs:-0.6,99.65) -- (axis cs:3.6,99.65);
      \nextgroupplot[width=0.575\linewidth, xmin=-0.6, xmax=5.6, /pgf/bar width=2.6pt, xtick={0,1,2,3,4,5}, xticklabels={Full,w/o $\mathcal{L}_1$,w/o $\mathcal{L}_2$,w/o $\mathcal{L}_3$,w/o $\mathcal{L}_4$,w/o $\mathcal{L}_5$}, title={(b) Objectives}]
        \addplot[fill=cPal] coordinates {(0,99.66) (1,90.42) (2,92.15) (3,72.43) (4,89.82) (5,86.51)};
        \addplot[fill=cAlt] coordinates {(0,99.64) (1,89.85) (2,90.84) (3,71.86) (4,90.34) (5,94.22)};
        \addplot[fill=cThr, postaction={pattern=north east lines, pattern color=white}] coordinates {(0,99.65) (1,90.13) (2,91.49) (3,72.14) (4,90.08) (5,90.19)};
        \draw[black!55, line width=0.5pt, dash pattern=on 2.5pt off 1.5pt] (axis cs:-0.6,99.65) -- (axis cs:5.6,99.65);
    \end{groupplot}
  \end{tikzpicture}

  \vspace{1pt}
  \pgfplotslegendfromname{ablegend}
  \vspace{-0.2cm}
  \caption{Ablation study on CIC-IDS2017. The dashed line marks the F1 score of the full model.}
  \label{res:abs}
  \vspace{-0.3cm}
\end{figure}
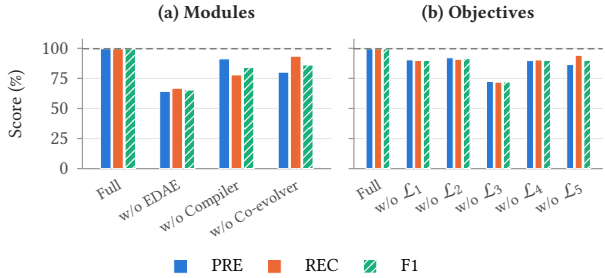
\addtocounter{figure}{1}

As shown in Figure~\ref{res:abs}, among the modules, removing EDAE causes the largest drop (F1 $\downarrow$34.16 points), as the model loses its hyper-rectangular explicit distribution boundaries and falls back to anomaly scores. Removing the compiler reduces recall by 21.75 points, since a global reconstruction threshold cannot separate attacks close to benign regions. Removing model refinement in the co-evolver reduces precision by 19.48 points when trained on Tuesday's benign traffic and tested on later days, as updating the model with the rules suppresses false positives from new applications.

Among the objectives, removing the density loss $\mathcal{L}_3$ causes the largest drop (F1 $\downarrow$27.51 points), as the rules no longer tightly enclose dense benign regions. Without $\mathcal{L}_4$, $\mathcal{L}_1$, and $\mathcal{L}_5$, F1 drops by 9.57, 9.52, and 9.46 points, as the bounds oscillate between collapse and explosion, lose consistency with the input, or take nearly uniform margins across features. Without $\mathcal{L}_2$, the rules become fragmented and overly narrow (F1 $\downarrow$8.16 points). Each objective is thus necessary, as they provide density ($\mathcal{L}_3$), generalization ($\mathcal{L}_2$), coherence ($\mathcal{L}_1$), stability ($\mathcal{L}_4$), and shape ($\mathcal{L}_5$).

\takeaway{Learning explicit benign boundaries hinges on balancing tightness and coverage, as removing the density term loosens the rules most, whereas removing the distribution term fragments them.}

\subsection{Deployment in Popular NIDS Engines (RQ5)}
\label{sec:online}
We integrate IDSForge into Suricata 8.0.1 under 10 Gbps traffic (Appendix~\ref{sec:onlineSetting}), translating the reduced rules automatically with the templates of \S\ref{sec:bre}. The conversion parses each rule into feature-interval pairs, omits inactive features, and encodes the bounds as Lua scripts, which are indexed by the R*-tree before loading.

As shown in Figure~\ref{res:online}, IDSForge keeps online F1 scores above 97.5\%, which are 0.69--2.11 percentage points below the offline results. The gap stems from flow features computed incrementally by the engine. ACC, AUC, and F1 stay close, so the online rules do not trade false positives for missed attacks. Compiling the 15,850 rules of CIC-IDS2017 takes 0.8 seconds, and loading them into Suricata takes 0.07 seconds. At runtime on CIC-IDS2017, IDSForge sustains 7.8 Gbps on a 16-core Intel Xeon Gold 5218 CPU with a median per-packet latency overhead of 3.1 $\mu$s, 0.31\% extra CPU utilization, and 42 MB of memory, less than 3\% of the typical footprint of Suricata. IDSForge can thus run on existing Web-facing gateways without additional hardware.

\takeaway{Compiling a learned detector into native engine rules lets it run on existing NIDS infrastructure at near line rate, and the remaining accuracy gap stems mainly from incremental feature computation in the engine.}

\section{Conclusion and Future Work}
\label{sec:conclusion}
We propose IDSForge, a self-explaining anomaly-based NIDS whose detector and explanation are the same set of executable rules. IDSForge trains an explicit distributional autoencoder to produce instance-specific hyper-rectangular explicit distribution boundaries of benign traffic rather than anomaly scores, compiles them into a compact rule base that popular NIDS engines execute natively, and regenerates the model and rules together when operator-confirmed false positives exceed a threshold. Experiments show that IDSForge reduces the F1 error of state-of-the-art anomaly-based detectors by $5.7\times$ to $47\times$, keeps its explanations consistent with detection, and sustains 7.8 Gbps in Suricata with a median per-packet latency overhead of 3.1 $\mu$s. Since each rule is a plain-text conjunction, operators can audit, edit, or revoke individual rules without retraining. These results suggest that using neural models to generate executable rules is a practical path toward interpretable learning-based detection for Web security.

Future work has four directions. First, we will study adaptive adversaries that probe or mimic the learned rules, building on the certified decision margin of Proposition~\ref{prop:margin}. Second, we will combine flow-level rules with payload inspection to detect exploits hidden in normal-looking sessions. Third, we will learn feature transformations that cover correlated features with fewer rules. We will also harden the co-evolver against manipulated false-positive feedback, which currently relies on trusted operator triage. Finally, we will evaluate IDSForge on large-scale production Web traffic and offload its rules to programmable switches~\cite{shi2025helios,li2024genos}.

\bibliographystyle{ACM-Reference-Format}
\bibliography{reference}

\appendix
\setcounter{dbltopnumber}{3}
\renewcommand{\dbltopfraction}{0.9}
\section{Theoretical Analysis}
\label{sec:theory}
This appendix provides theoretical support for the three modules of IDSForge (\S\ref{sec:design}). Table~\ref{tab:theory} summarizes the four results, each backing a design claim in the main text.

\begin{table}[t]
  \centering
  \caption{Overview of the theoretical results.}
  \label{tab:theory}
  \tabstyle
  \begin{tabularx}{\linewidth}{@{\hspace{4pt}}l X l@{\hspace{4pt}}}
    \tabtop
    \rowcolor{tabHead} \textbf{Result} & \textbf{Guarantee} & \textbf{Module} \\
    \tabmid
    Prop.~\ref{prop:density}   & Rules tightly enclose benign flows, and the smooth count is accurate & EDAE \\
    Prop.~\ref{prop:reduction} & Rule reduction preserves the benign region & Compiler \\
    Prop.~\ref{prop:margin}    & Explanations are faithful and certify a decision margin & Compiler \\
    Cor.~\ref{prop:forgetting} & Updates never reject previously accepted flows & Co-evolver \\
    \tabbot
  \end{tabularx}
  \vspace{-0.3cm}
\end{table}

\noindent\textbf{Notation.} Flow features are min-max normalized to $\mathcal{X}=[0,1]^n$, and test inputs are clipped to $\mathcal{X}$. A rule is a closed hyper-rectangle $H=\prod_{j=1}^{n}[l_j,u_j]$, where an inactive feature $j\notin\mathcal{A}_H$ takes $[l_j,u_j]=[0,1]$. For a finite rule base $\mathcal{H}$, the acceptance region and the detector are
\begin{equation*}
R(\mathcal{H})=\bigcup_{H\in\mathcal{H}}H,\qquad f_{\mathcal{H}}(x)=\mathbb{1}\big[x\notin R(\mathcal{H})\big].
\end{equation*}
For a rule $H$ and a flow $x$, the violation set and the violation distance are
\begin{gather*}
V(x,H)=\{j\in\mathcal{A}_H\mid x_j\notin[l_j,u_j]\},\\
d_\infty(x,H)=\max_{j}\,\mathrm{dist}\big(x_j,[l_j,u_j]\big).
\end{gather*}
We write $\mathrm{vol}(\cdot)$ for the Lebesgue measure on $\mathcal{X}$.

\subsection{Explicit Distributional Autoencoder}
\label{sec:theory-bae}
The density loss $\mathcal{L}_3$ is meant to make each rule enclose its benign flows tightly. Proposition~\ref{prop:density} shows that its minimizer is exactly such a box, and that replacing the non-differentiable count by the smooth count of \S\ref{sec:bae} introduces only a small error.

\begin{proposition}[Tightness of the density objective]
\label{prop:density}
\leavevmode\par\noindent\textup{(a)} Fix a mini-batch $\mathcal{B}$ and a rule $H$. Every minimizer of $\mathrm{vol}(H)/(N_{\mathrm{in}}(H)+\varepsilon)$ is the bounding box of the samples it contains.
\par\noindent\textup{(b)} Let $k=|\mathcal{A}_H|$, and let $\tilde N_{\mathrm{in}}$ be the smooth count with temperature $\tau$. If every $x'\in\mathcal{B}$ lies at $\ell_\infty$-distance at least $\Delta>0$ from the boundary of $H$, then
\begin{equation}
|\tilde N_{\mathrm{in}}-N_{\mathrm{in}}|\le 2k|\mathcal{B}|\,e^{-\Delta/\tau}.
\end{equation}
\end{proposition}
\begin{proof}
(a) The bounding box of $H\cap\mathcal{B}$ lies inside $H$ and contains the same samples, so it keeps the count without increasing the volume.
\par\noindent(b) We use $\sigma(-t)\le e^{-t}$ for $t\ge0$. A sample inside $H$ with margin $\Delta$ has $2k$ sigmoid factors of at least $1-e^{-\Delta/\tau}$, so by Bernoulli's inequality its product lies in $[1-2ke^{-\Delta/\tau},1]$. A sample outside $H$ violates some bound by at least $\Delta$, so its product is at most $e^{-\Delta/\tau}$. Summing over $\mathcal{B}$ gives the bound.
\end{proof}

\noindent\textit{Remark.} With the default $\tau=0.05$, a margin of $\Delta=0.3$ already gives $e^{-\Delta/\tau}<0.003$, so the smooth count closely tracks the exact count used at deployment.

\subsection{High-Speed Rule Compiler and Explanation}
\label{sec:theory-bre}
The compiler must shrink the rule base without changing what the detector accepts. Proposition~\ref{prop:reduction} guarantees this for the Boolean minimization of \S\ref{sec:bre}.

\begin{proposition}[Lossless rule reduction]
\label{prop:reduction}
Let $\mathcal{H}_s$ be the rules of a server, and let $\mathcal{H}'_s$ be the hyper-rectangles decoded from the sum of products returned by Boolean minimization over the threshold variables of \S\ref{sec:bre}. Then $R(\mathcal{H}'_s)=R(\mathcal{H}_s)$ up to a set of measure zero on cell boundaries.
\end{proposition}
\begin{proof}
Since the cut points contain all rule bounds, they split $\mathcal{X}$ into cells, and each rule is a union of cells. Each cell corresponds to one order-respecting assignment of the threshold variables, on which $F$ equals $1$ if and only if the cell lies in $R(\mathcal{H}_s)$. As only order-violating assignments are don't-cares, the minimized $F'$ agrees with $F$ on every cell. A product term of $F'$ constrains each feature to a contiguous range of cells, so it decodes to the hyper-rectangle formed by the cells that satisfy it. Hence $R(\mathcal{H}'_s)$ is the union of the cells with $F'=1$, which is $R(\mathcal{H}_s)$.
\end{proof}

\noindent\textbf{Complexity.} Rule generation needs one forward pass per flow. For $K$ rules of a server, a point query on an R*-tree with fanout $M$ visits $O(\log_M K)$ nodes under low overlap and $O(K)$ nodes in the worst case. Minimization reduces $K$, and the R*-tree reduces overlap.

\smallskip
\noindent\textbf{Faithfulness.} Since inactive features always hold on $\mathcal{X}$, a flow satisfies a rule $H$ exactly when $V(x,H)=\emptyset$. Hence $f_{\mathcal{H}}(x)=1$ if and only if every rule is violated, and the explanation reported by IDSForge, i.e., the matched rule of a benign flow or the violated conditions of the nearest rule of a malicious flow, is a necessary and sufficient certificate of the decision. In contrast, a post-hoc surrogate generally disagrees with the detector on some flows~\cite{li2023interpreting,han2024rules}.

Beyond faithfulness, the explanation also quantifies robustness. Proposition~\ref{prop:margin} shows that the distance to the nearest rule is the exact perturbation budget needed to change a decision. We define the distance of a flow to the rule base and its margin inside an accepting rule as
\begin{equation}
\eta(x)=\min_{H\in\mathcal{H}}d_\infty(x,H),\qquad
\gamma(x)=\max_{H\ni x}\,\min_{j\in\mathcal{A}_H}\min(x_j-l_j,\,u_j-x_j).
\end{equation}

\begin{proposition}[Certified decision margin]
\label{prop:margin}
\leavevmode\par\noindent\textup{(a)} If $f_{\mathcal{H}}(x)=1$, then $\eta(x)>0$, and $f_{\mathcal{H}}(x+\delta)=1$ for every $\delta$ with $\|\delta\|_\infty<\eta(x)$. This bound is tight, i.e., some $\delta$ with $\|\delta\|_\infty=\eta(x)$, supported only on $V(x,H^*)$ for the nearest rule $H^*$, yields $f_{\mathcal{H}}(x+\delta)=0$.
\par\noindent\textup{(b)} If $f_{\mathcal{H}}(x)=0$ and $\gamma(x)>0$, then $f_{\mathcal{H}}(x+\delta)=0$ for every $\delta$ with $\|\delta\|_\infty<\gamma(x)$ and $x+\delta\in\mathcal{X}$.
\end{proposition}
\begin{proof}
(a) As each $H$ is closed and $\mathcal{H}$ is finite, $\eta(x)>0$. Since $d_\infty(\cdot,H)$ is $1$-Lipschitz in $\|\cdot\|_\infty$, we have $d_\infty(x+\delta,H)\ge d_\infty(x,H)-\|\delta\|_\infty>0$ for each $H$. For tightness, project each violated feature $j\in V(x,H^*)$ onto $[l^*_j,u^*_j]$ and keep the other features unchanged, which moves $x$ into $H^*$ with $\|\delta\|_\infty=\eta(x)$.
\par\noindent(b) For the rule $H$ attaining $\gamma(x)$, every active feature satisfies $l_j<x_j+\delta_j<u_j$, and inactive features stay in $[0,1]$, so $x+\delta\in H$.
\end{proof}

\noindent\textit{Remark.} Benign flows with margin $\gamma$ are thus robust to noise smaller than $\gamma$ (Figure~\ref{fig:intro1_b}), and an attacker must change the violated features of a malicious flow by at least $\eta(x)$ to evade detection. Since the explanation names these features and the required changes, operators can directly assess the evasion risk of an alert.

\subsection{Self-Distillation Co-Evolver}
\label{sec:theory-blu}
An update should extend the benign region to new behaviors without dropping what was learned before. This follows directly from the lossless reduction.

\begin{corollary}[No rule-level forgetting]
\label{prop:forgetting}
Let $\mathcal{H}_{t+1}$ be obtained from $\mathcal{H}_t$ by merging the regenerated rules $\mathcal{H}_{\mathrm{new}}$ and applying the reduction of the compiler. Then $R(\mathcal{H}_t)\subseteq R(\mathcal{H}_{t+1})$, i.e., flows accepted before the update remain accepted after it, and the region grows by at most $\sum_{H\in\mathcal{H}_{\mathrm{new}}}\mathrm{vol}(H)$.
\end{corollary}
\begin{proof}
Merging only adds rules, and the reduction preserves their union by Proposition~\ref{prop:reduction}, so $R(\mathcal{H}_{t+1})=R(\mathcal{H}_t)\cup\bigcup_{H\in\mathcal{H}_{\mathrm{new}}}H$.
\end{proof}

\noindent\textit{Remark.} The region growth, which may admit attacks, is bounded by the volume of the new rules and occurs only when the false-positive rate exceeds $\theta_{\mathrm{FP}}$.

\section{Rule Compilation Algorithm}
\label{sec:alg}
Algorithm~\ref{alg:bre} summarizes the per-server compilation of \S\ref{sec:bre}.

\begin{algorithm}[t]
  \caption{Per-server rule compilation}
  \label{alg:bre}
  \small
  \SetKwProg{Fn}{Function}{}{}
  \SetKwFunction{Esp}{Espresso}\SetKwFunction{Str}{STR}
  \KwIn{$\mathcal{D}=\bigcup_{s}\mathcal{D}_s$, EDAE $g_\theta$, engine $E$, capacity $M$}
  \KwOut{$\{(\mathcal{H}'_s,\mathcal{T}_s,\Pi_s)\}_s$}
  \ForEach{server $s$}{
    $\mathcal{H}_s\leftarrow\big\{\prod_{j\in\mathcal{A}(x)}[l_j(x),u_j(x))\,\big|\,x\in\mathcal{D}_s\big\}$ \tcp*{hyper-rectangles}
    $F_s\leftarrow\bigvee_{H\in\mathcal{H}_s}\bigwedge_{j\in\mathcal{A}_H}\big(b_{j,\kappa_j(l_{Hj})}\wedge\neg b_{j,\kappa_j(u_{Hj})}\big)$ \tcp*{encode}
    $\mathcal{H}'_s\leftarrow\{\mathrm{dec}(p)\mid p\in\Esp(F_s,\mathcal{C})\}$ \tcp*{reduce}
    $\mathcal{T}_s\leftarrow\Str(\mathcal{H}'_s,M)$ \tcp*{index}
    $\Pi_s\leftarrow\{\tau_E(H,s)\}_{H\in\mathcal{H}'_s}\cup\{\alpha_E(s)\}$ \tcp*{translate}
  }
  \Fn{\Esp{$F,\mathcal{C}$}}{
    $\mathcal{R}\leftarrow\neg(F\vee\mathcal{C})$;\ $F'\leftarrow\textsc{Irr}(\textsc{Exp}(F))$ \tcp*{OFF-set, primes}
    \lRepeat{$|F'|$ \emph{is unchanged}}{$F'\leftarrow\textsc{Irr}(\textsc{Exp}(\textsc{Red}(F')))$}
    \Return $F'$ \tcp*{shrink, grow, drop}
  }
  \Fn{\Str{$\mathcal{H},M$}}{
    $P\leftarrow\lceil|\mathcal{H}|/M\rceil$, $S\leftarrow\lceil P^{1/n}\rceil$; tile centers into $S$ slabs\;
    Pack leaves of size $M$, then $\{B(N)\}$ level by level; \Return root\;
  }
\end{algorithm}

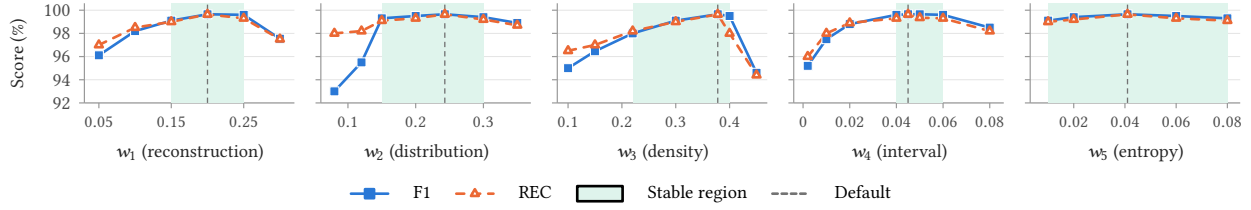
\begin{figure*}[t]
  \centering
  \begin{tikzpicture}
    \begin{groupplot}[
      paperstyle, linelegend,
      group style={group size=5 by 1, horizontal sep=0.35cm, y descriptions at=edge left},
      width=0.245\textwidth, height=2.9cm,
      ymin=92, ymax=100.6, ytick={92,94,96,98,100},
      scaled x ticks=false, x tick label style={font=\scriptsize, /pgf/number format/fixed, /pgf/number format/precision=2},
      every axis plot/.append style={line width=1pt, mark size=1.7pt},
    ]
      \nextgroupplot[xlabel={$w_1$ (reconstruction)}, xmin=0.03, xmax=0.32, xtick={0.05,0.15,0.25}, legend to name=senslegend, ylabel={Score (\%)}]
        \fill[cThr!14] (axis cs:0.15,92) rectangle (axis cs:0.25,100.6);
        \draw[black!55, dash pattern=on 2pt off 1.5pt, line width=0.5pt] (axis cs:0.20,92) -- (axis cs:0.20,100.6);
        \addplot[cPal, mark=square*] coordinates {(0.05,96.11) (0.1,98.2) (0.15,99.1) (0.2,99.65) (0.25,99.6) (0.3,97.54)};
        \addlegendentry{F1}
        \addplot[cAlt, dashed, mark=triangle*, mark options={solid, fill=white}] coordinates {(0.05,97) (0.1,98.5) (0.15,99) (0.2,99.64) (0.25,99.3) (0.3,97.5)};
        \addlegendentry{REC}
        \addlegendimage{area legend, fill=cThr!14, draw=none}\addlegendentry{Stable region}
        \addlegendimage{black!55, dash pattern=on 2pt off 1.5pt}\addlegendentry{Default}
      \nextgroupplot[xlabel={$w_2$ (distribution)}, xmin=0.06, xmax=0.37, xtick={0.1,0.2,0.3}]
        \fill[cThr!14] (axis cs:0.15,92) rectangle (axis cs:0.30,100.6);
        \draw[black!55, dash pattern=on 2pt off 1.5pt, line width=0.5pt] (axis cs:0.243,92) -- (axis cs:0.243,100.6);
        \addplot[cPal, mark=square*] coordinates {(0.08,93) (0.12,95.5) (0.15,99.3) (0.2,99.5) (0.243,99.65) (0.3,99.4) (0.35,98.9)};
        \addplot[cAlt, dashed, mark=triangle*, mark options={solid, fill=white}] coordinates {(0.08,98) (0.12,98.2) (0.15,99.1) (0.2,99.3) (0.243,99.64) (0.3,99.2) (0.35,98.7)};
      \nextgroupplot[xlabel={$w_3$ (density)}, xmin=0.08, xmax=0.47, xtick={0.1,0.2,0.3,0.4}]
        \fill[cThr!14] (axis cs:0.22,92) rectangle (axis cs:0.40,100.6);
        \draw[black!55, dash pattern=on 2pt off 1.5pt, line width=0.5pt] (axis cs:0.378,92) -- (axis cs:0.378,100.6);
        \addplot[cPal, mark=square*] coordinates {(0.1,95) (0.15,96.46) (0.22,98) (0.3,99.1) (0.378,99.65) (0.4,99.5) (0.45,94.6)};
        \addplot[cAlt, dashed, mark=triangle*, mark options={solid, fill=white}] coordinates {(0.1,96.5) (0.15,97) (0.22,98.2) (0.3,99) (0.378,99.64) (0.4,98) (0.45,94.4)};
      \nextgroupplot[xlabel={$w_4$ (interval)}, xmin=-0.004, xmax=0.086, xtick={0,0.02,0.04,0.06,0.08}]
        \fill[cThr!14] (axis cs:0.04,92) rectangle (axis cs:0.06,100.6);
        \draw[black!55, dash pattern=on 2pt off 1.5pt, line width=0.5pt] (axis cs:0.045,92) -- (axis cs:0.045,100.6);
        \addplot[cPal, mark=square*] coordinates {(0.002,95.2) (0.01,97.5) (0.02,98.8) (0.04,99.6) (0.045,99.65) (0.05,99.65) (0.06,99.6) (0.08,98.5)};
        \addplot[cAlt, dashed, mark=triangle*, mark options={solid, fill=white}] coordinates {(0.002,96) (0.01,98) (0.02,98.9) (0.04,99.3) (0.045,99.64) (0.05,99.35) (0.06,99.3) (0.08,98.2)};
      \nextgroupplot[xlabel={$w_5$ (entropy)}, xmin=0.003, xmax=0.085, xtick={0.02,0.04,0.06,0.08}]
        \fill[cThr!14] (axis cs:0.01,92) rectangle (axis cs:0.08,100.6);
        \draw[black!55, dash pattern=on 2pt off 1.5pt, line width=0.5pt] (axis cs:0.041,92) -- (axis cs:0.041,100.6);
        \addplot[cPal, mark=square*] coordinates {(0.01,99.1) (0.02,99.4) (0.041,99.65) (0.06,99.5) (0.08,99.3)};
        \addplot[cAlt, dashed, mark=triangle*, mark options={solid, fill=white}] coordinates {(0.01,99) (0.02,99.2) (0.041,99.64) (0.06,99.3) (0.08,99.1)};
    \end{groupplot}
  \end{tikzpicture}

  \vspace{1pt}
  \pgfplotslegendfromname{senslegend}
  \vspace{-0.2cm}
  \caption{Single-weight sensitivity on CIC-IDS2017. Shaded areas mark stable regions, and dashed lines mark default values.}
  \label{fig:sens-single}
\end{figure*}

\begin{figure*}[t]
  \centering
  \begin{tikzpicture}
    \begin{axis}[
      paperstyle,
      width=0.36\textwidth, height=3.3cm,
      title={(a) Interaction of $w_2$ and $w_3$ (F1, \%)},
      xlabel={$w_3$ (density)}, ylabel={$w_2$ (distribution)},
      xmin=-0.5, xmax=4.5, ymin=-0.5, ymax=4.5, enlargelimits=false,
      xtick={0,...,4}, xticklabels={0.15,0.25,0.378,0.42,0.45},
      ytick={0,...,4}, yticklabels={0.12,0.18,0.243,0.30,0.35},
      ymajorgrids=false, axis lines*=box, tick style={draw=none},
      colormap={sens}{color=(white) color=(cPal!30) color=(cPal!75)},
      point meta min=92, point meta max=99.7,
      colorbar, colorbar style={width=2.5mm, font=\scriptsize, ytick={92,94,96,98}},
      nodes near coords={\pgfmathprintnumber[fixed,precision=2]{\pgfplotspointmeta}},
      every node near coord/.append style={font=\scriptsize, anchor=center, text=black!85},
    ]
      \addplot[matrix plot*, mesh/cols=5, point meta=explicit] table[meta=z] {
x y z
0 0 93.5
1 0 95.2
2 0 96.8
3 0 95
4 0 92.1
0 1 95.2
1 1 97.8
2 1 98.9
3 1 98.2
4 1 95.5
0 2 97
1 2 98.6
2 2 99.65
3 2 99.2
4 2 97.1
0 3 96.1
1 3 97.5
2 3 99.1
3 3 98.5
4 3 96
0 4 94
1 4 95.8
2 4 97.5
3 4 96.8
4 4 93.5
      };
    \end{axis}
  \end{tikzpicture}\hfill
  \begin{tikzpicture}
    \begin{groupplot}[
      paperstyle, linelegend,
      group style={group size=2 by 1, horizontal sep=1.1cm},
      width=0.3\textwidth, height=3.3cm, xmin=0.5, xmax=15.5, xtick={1,5,10,15},
      xlabel={Epoch},
      every axis plot/.append style={line width=1pt, mark size=1.6pt},
    ]
      \nextgroupplot[ylabel={F1 (\%)}, ymin=64, ymax=101, title={(b) Convergence (F1)}, legend to name=convlegend]
        \addplot[cAlt, densely dotted, mark=*, mark options={solid, fill=white}] coordinates {(1,66.89) (2,72.53) (3,77.15) (4,80.93) (5,84.02) (6,86.55) (7,88.63) (8,90.33) (9,91.72) (10,92.86) (11,93.79) (12,94.55) (13,95.18) (14,95.69) (15,96.11)};
        \addlegendentry{$w_1=0.05$}
        \addplot[cPal, mark=square*] coordinates {(1,78.41) (2,84.43) (3,88.75) (4,91.84) (5,94.06) (6,95.65) (7,96.78) (8,97.6) (9,98.18) (10,98.6) (11,99.0) (12,99.25) (13,99.43) (14,99.56) (15,99.65)};
        \addlegendentry{$w_1=0.20$ (default)}
        \addplot[cThr, dashed, mark=triangle*, mark options={solid, fill=white}] coordinates {(1,74.04) (2,80.59) (3,85.34) (4,88.77) (5,91.26) (6,93.07) (7,94.37) (8,95.32) (9,96) (10,96.5) (11,96.86) (12,97.12) (13,97.3) (14,97.44) (15,97.54)};
        \addlegendentry{$w_1=0.30$}
      \nextgroupplot[ylabel={Training loss}, ymin=0, ymax=0.72, title={(c) Convergence (loss)}]
        \addplot[cAlt, densely dotted, mark=*, mark options={solid, fill=white}] coordinates {(1,0.6959) (2,0.6032) (3,0.4881) (4,0.4461) (5,0.3893) (6,0.3035) (7,0.2707) (8,0.2179) (9,0.2013) (10,0.1768) (11,0.1583) (12,0.1557) (13,0.1069) (14,0.0996) (15,0.0674)};
        \addplot[cPal, mark=square*] coordinates {(1,0.6202) (2,0.4549) (3,0.3568) (4,0.2555) (5,0.196) (6,0.1502) (7,0.1141) (8,0.0872) (9,0.0768) (10,0.0578) (11,0.0379) (12,0.0453) (13,0.0215) (14,0.02) (15,0.0292)};
        \addplot[cThr, dashed, mark=triangle*, mark options={solid, fill=white}] coordinates {(1,0.6804) (2,0.4864) (3,0.3705) (4,0.287) (5,0.1941) (6,0.1538) (7,0.1193) (8,0.0849) (9,0.06) (10,0.0409) (11,0.0309) (12,0.0454) (13,0.036) (14,0.03) (15,0.03)};
    \end{groupplot}
  \end{tikzpicture}

  \vspace{1pt}
  \hspace*{0.36\textwidth}\pgfplotslegendfromname{convlegend}
  \vspace{-0.2cm}
  \caption{Weight interaction and convergence on CIC-IDS2017. (a) F1 over $w_2$ and $w_3$. (b, c) F1 and training loss per epoch for three values of $w_1$.}
  \label{fig:sens-inter}
\end{figure*}
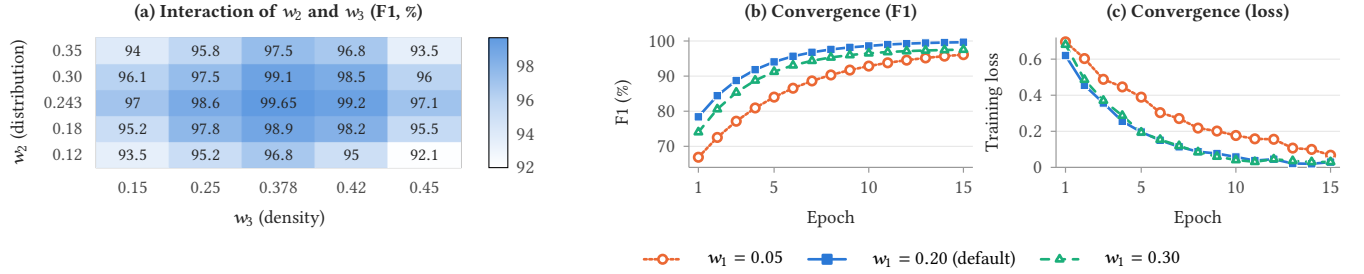

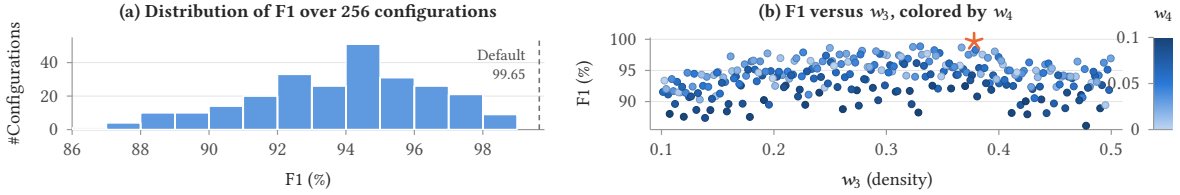
\begin{figure*}[t]
  \centering
  \begin{tikzpicture}
    \begin{groupplot}[
      paperstyle,
      group style={group size=2 by 1, horizontal sep=1.4cm},
      width=0.44\textwidth, height=2.8cm,
    ]
      \nextgroupplot[title={(a) Distribution of F1 over 256 configurations}, xlabel={F1 (\%)}, ylabel={\#Configurations}, ymin=0, ymax=55, xmin=86, xmax=99.8, xtick={86,88,90,92,94,96,98}]
        \addplot[hist={bins=13, data min=86, data max=99}, fill=cPal!75, draw=white, line width=0.6pt] table[y=F1] {fig/sobol256.dat};
        \draw[black!55, dash pattern=on 2pt off 1.5pt, line width=0.6pt] (axis cs:99.65,0) -- (axis cs:99.65,55);
        \node[font=\scriptsize, text=black!70, anchor=north east, align=right] at (axis cs:99.5,53) {Default\\99.65};
      \nextgroupplot[title={(b) F1 versus $w_3$, colored by $w_4$}, xlabel={$w_3$ (density)}, ylabel={F1 (\%)}, ymin=85.5, ymax=100.3, xmin=0.09, xmax=0.51,
          colormap={sens4}{color=(cPal!30) color=(cPal) color=(cPal!55!black)}, point meta min=0, point meta max=0.1,
          colorbar, colorbar style={width=2.5mm, font=\scriptsize, ytick={0,0.05,0.1}, scaled y ticks=false, yticklabel style={/pgf/number format/fixed}, title={$w_4$}, title style={font=\scriptsize}}]
        \addplot[scatter, only marks, mark=*, mark size=1.3pt, scatter src=explicit, mark options={draw=white, line width=0.2pt}] table[x=w3, y=F1, meta=w4] {fig/sobol256.dat};
        \addplot[only marks, mark=star, mark size=3.2pt, mark options={draw=cAlt, line width=1pt}] coordinates {(0.378,99.65)};
    \end{groupplot}
  \end{tikzpicture}
  \vspace{-0.2cm}
  \caption{Global sensitivity over 256 Sobol-sampled weight configurations on CIC-IDS2017. The dashed line and the star mark the default.}
  \label{fig:sens-sobol}
\end{figure*}

\section{Experimental Setup}
\label{sec:setup}

\subsection{Description of Baselines}
\label{sec:baselines}
Table~\ref{tab:baselines} summarizes the baselines.

\noindent\textbf{Detectors.} \emph{ET-BERT}~\cite{lin2022bert} pre-trains a Transformer on unlabeled encrypted traffic for classification. \emph{TMG-GAN}~\cite{ding2024tmggan} balances rare attack classes with samples from multiple generators. \emph{Whisper}~\cite{fu2021realtime} clusters benign frequency-domain features and flags distant flows. \emph{HyperVision}~\cite{fu2023detecting} detects abnormal structures in a flow interaction graph. \emph{Kitsune}~\cite{mirsky2018kitsune} uses the reconstruction errors of an autoencoder ensemble. \emph{Trident}~\cite{zhao2024trident} runs per-class one-class learners and clusters rejected outliers into new classes. \emph{UnDiff}~\cite{lian2025facing} scores inter-sample uncertainty deviations estimated by evidential learning. \emph{Diff-RF}~\cite{marteau2021random} scores leaf distances and visiting frequencies in random partitioning trees. \emph{ContraMTD}~\cite{han2024contramtd} contrasts local behavior and global interaction features of benign traffic.

\noindent\textbf{Explainers.} All four explain the same autoencoder (AE) trained on benign flows and access only benign flows and the binary AE decisions (\S\ref{sec:interp}). \emph{LIME}~\cite{ribeiro2016why} and \emph{LEMNA}~\cite{guo2018lemna} fit a sparse linear model and a fused-lasso mixture regression around a flow, aggregated by submodular pick. \emph{Trustee}~\cite{jacobs2022ai} distills the AE into a high-fidelity decision tree. \emph{SCD-Tree}~\cite{zhang2024dissect} segments the data by the AE outputs and fits a Gaussian boundary per segment as interval rules.

\begin{table}[t]
  \centering
  \caption{Baselines and IDSForge in our setup. AE denotes the explained autoencoder, and the shaded row is IDSForge.}
  \label{tab:baselines}
  \tabstyle
  \setlength{\tabcolsep}{2pt}
  \begin{tabular*}{\linewidth}{@{\hspace{\tabcolsep}\extracolsep{\fill}}lllll}
    \tabtop
    \rowcolor{tabHead} \textbf{Method} & Type & Input & Training data & Output \\
    \tabmid
    ET-BERT     & Supervised & Packet bytes        & Benign + attack & Class label \\
    TMG-GAN     & Supervised & Flow features       & Benign + attack & Class label \\
    \tabmid
    Whisper     & Anomaly    & Frequency features  & Benign & Cluster distance \\
    HyperVision & Anomaly    & Interaction graph   & Benign & Structural score \\
    Kitsune     & Anomaly    & Packet statistics   & Benign & Recon. error \\
    Trident     & Anomaly    & Flow features       & Benign & One-class score \\
    UnDiff      & Anomaly    & Flow features       & Benign & Uncertainty \\
    Diff-RF     & Anomaly    & Flow features       & Benign & Path-length score \\
    ContraMTD   & Anomaly    & Interaction views   & Benign & Contrastive score \\
    \tabmid
    LIME        & Local      & AE decisions        & Benign & Feature weights \\
    LEMNA       & Local      & AE decisions        & Benign & Feature weights \\
    Trustee     & Global     & AE decisions        & Benign & Decision tree \\
    SCD-Tree    & Global     & AE decisions        & Benign & Interval rules \\
    \tabmid
    \rowcolor{tabOurs} \textbf{IDSForge} & Self-explaining & Flow features & Benign & Executable rules \\
    \tabbot
  \end{tabular*}
\end{table}

\subsection{Implementation and Deployment}
\label{sec:onlineSetting}
Each flow has $n=8$ packet-, flow-, and session-level features on timing, size, and topology, extracted from the packet captures instead of the released features (Table~\ref{tab:datasets}). EDAE has an isomorphic latent space and a 4-layer encoder and decoder with 256 hidden units, trained with Adam (learning rate 0.001, batch size 1024) for 15 epochs and $\tau=0.05$. The R*-tree has capacity $M=32$, and the co-evolver uses $\theta_{\mathrm{FP}}=0.01$ with 64 samples per hyper-rectangle.
We deploy IDSForge with Suricata 8.0.1 on a server with a 16-core Intel Xeon Gold 5218 CPU, a 10 GbE interface, and 128 GB DDR4 RAM. For CIC-IDS2017, CSE-CIC-IDS2018, and CIC-UNSW-NB15, the test-set packet captures are replayed with tcpreplay at up to 10 Gbps and processed in DPDK zero-copy mode, with 15,850 rules for CIC-IDS2017.

\section{Additional Experiments}
\label{sec:moreexp}

\subsection{Comparison of Learning Paradigms}
\label{sec:appendix1}
To choose the geometry of the benign region, we compare six paradigms on the 8 features of CIC-IDS2017 with a random 9:1 split (Table~\ref{table:motivationConstract}). Tree models achieve the highest F1 with fast inference, distance-based models with spherical regions rank second, and linear and neural models perform worse. We use CART (Gini), L2 logistic regression (100 iterations), KNN ($k=5$), LinearSVC ($C=1$, as RBF exceeded 2 hours), Gaussian naive Bayes, and a Transformer (8 heads, 512 hidden units, 30 epochs, other settings as EDAE).

\takeaway{Since each flow feature has a direct meaning and its benign values fall in limited ranges, axis-aligned boundaries fit benign traffic well, so interpretability need not cost accuracy.}

\begin{table}[t]
  \centering
  \caption{Learning paradigms on CIC-IDS2017 (metrics in \%, times in seconds). The shaded row is the best paradigm.}
  \label{table:motivationConstract}
  \tabstyle
  \setlength{\tabcolsep}{2.5pt}
  \begin{tabular}{lcccccc}
    \tabtop
    \rowcolor{tabHead} \textbf{Paradigm} & ACC & PRE & REC & F1 & Train(s) & Pred(s) \\
    \tabmid
    \rowcolor{tabOurs} Tree (CART)~\cite{lewis2000introduction} & 99.8 & 99.5 & 99.4 & 99.4 & 150.82 & 0.22 \\
    Additive (LR)~\cite{healy2006logistic} & 93.4 & 83.1 & 83.2 & 83.2 & 31.44 & 0.08 \\
    Distance (KNN)~\cite{samet2007k} & 99.3 & 97.4 & 98.8 & 98.1 & 0.30 & 652.23 \\
    Kernel (SVM)~\cite{chauhan2019problem} & 93.8 & 83.6 & 85.3 & 84.4 & 154.91 & 0.09 \\
    Probabilistic (NB)~\cite{heckerman1997bayesian} & 47.4 & 27.0 & 97.5 & 42.2 & 4.52 & 1.67 \\
    Neural (Transformer)~\cite{vaswani2017attention} & 98.5 & 94.6 & 98.0 & 96.3 & 279.54 & 27.54 \\
    \tabbot
  \end{tabular}
  \vspace{-0.3cm}
\end{table}

\subsection{Sensitivity Analysis}
\label{sec:sensitivity}
\noindent\textbf{Single-weight sweeps.} On CIC-IDS2017, each weight (default 0.20, 0.243, 0.378, 0.045, and 0.041 for $w_1$ to $w_5$) has a stable region with F1 above 98\% (Figure~\ref{fig:sens-single}). (\rmnum{1}) $w_3$ matters most, as raising it from 0.15 to 0.378 adds 3.2 points of F1, whereas 0.45 over-contracts the rules (REC $\downarrow$5.2 points). (\rmnum{2}) $w_2\le0.12$ fragments the hyper-rectangles (F1 at most 95.5\%). (\rmnum{3}) $w_1=0.30$ weakens the geometric objectives (F1 $\downarrow$2.11 points). (\rmnum{4}) $w_4\le0.02$ lets the bounds oscillate (F1 as low as 95.2\%). (\rmnum{5}) $w_5$ changes F1 by less than 0.6 points in $[0.01,0.08]$, but removing it lowers F1 to 90.19\% (\S\ref{sec:ablation}).

\noindent\textbf{Interaction and convergence.} In Figure~\ref{fig:sens-inter}(a), F1 peaks at the default (99.65\%), and a small $w_2$ with a large $w_3$ over-contracts the rules (92.1\%) while the reverse leaves them loose (94.0\%). In Figure~\ref{fig:sens-inter}(b, c), the default $w_1$ reaches 99.65\% F1 at epoch 15, whereas 0.05 converges slowly (96.11\%) and 0.30 plateaus at 97.54\%.

\noindent\textbf{Global sensitivity.} Over 256 Sobol-sampled combinations (Figure~\ref{fig:sens-sobol}), F1 spans 86.12\%--98.87\% (median 94.16\%), below the default. The 30 combinations reaching 97\% use $w_4\le0.06$ and $w_3\le0.40$, whereas the 25 below 90\% mostly have a large $w_4$ or a small $w_2$ (medians 0.089 and 0.116).

\takeaway{The density and distribution terms pull the bounds in opposite directions, so their balance matters most.}

\end{document}